\documentclass[13pt]{article}

\usepackage{fullpage}
\usepackage{graphicx}
\usepackage{amssymb}
\usepackage{amsmath}
\usepackage{enumitem}
\usepackage{xcolor}

\usepackage[T1]{fontenc}
\DeclareFontFamily{T1}{calligra}{}
\DeclareFontShape{T1}{calligra}{m}{n}{<->s*[1.04]callig15}{}
\DeclareMathAlphabet\mathcalligra   {T1}{calligra} {m} {n}
\DeclareMathAlphabet\mathzapf       {T1}{pzc} {mb} {it}
\DeclareMathAlphabet\mathchorus     {T1}{qzc} {m} {n}
\DeclareMathAlphabet\mathrsfso      {U}{rsfso}{m}{n}

\numberwithin{equation}{section}
\newcommand{\e}{\mathrm{e}}

\newcommand{\C}{\mathbb{C}}

\newcommand{\R}{\mathbb{R}}

\newcommand{\Ra}{R_{\Omega}}

\newcommand{\Hm}[1]{\leavevmode{\marginpar{\tiny%
$\hbox to 0mm{\hspace*{-0.5mm}$\leftarrow$\hss}%
\vcenter{\vrule depth 0.1mm height 0.1mm width \the\marginparwidth}%
\hbox to
0mm{\hss$\rightarrow$\hspace*{-0.5mm}}$\\\relax\raggedright #1}}}

\newtheorem{claim}{Claim}[section]
\newtheorem{theorem}[claim]{Theorem}
\newtheorem{lemma}[claim]{Lemma}

\newtheorem{remark}[claim]{Remark}

\newenvironment{proof}[1][Proof]{\textsl{#1.} }{\ \rule{0.4em}{0.7em}}

\begin{document}

\title{ Fermi's golden rule in  tunneling models with quantum waveguides perturbed by Kato class measures}

\author{Sylwia Kondej\footnote{Corresponding author.}  \,  and  Kacper Ślipko}
\date{\small 
Institute of Physics, University of Zielona G\'ora, ul.\ Szafrana 4a, \\ 65246 Zielona G\'ora, Poland \\ \emph{e-mail:  s.kondej@if.uz.zgora.pl\,,  slipkokacper@gmail.com}
}

\maketitle

\begin{abstract} In this paper we consider  two dimensional quantum system with an infinite waveguide of the width $d$ and  a transversally invariant profile. Furthermore, we assume that at a distant $\rho$ there is a perturbation defined by the Kato measure. We show that, under certain conditions, the resolvent of the Hamiltonian has the second sheet pole which reproduces the resonance at $z(\rho)$ with the asymptotics $z(\rho)=\mathcal E_{\beta ; n}+\mathcal O \Big(\frac{ \exp(-\sqrt{2 |\mathcal E_{\beta ;n}| } \rho  )}{\rho }\Big)$ for $\rho$ large and with the resonant energy $\mathcal E_{\beta ;n}$. Moreover, we show that the imaginary component of $z(\rho)$ satisfies  Fermi's golden  rule which we explicitly derive.
\end{abstract}

\bigskip

\medskip

\section{Introduction}


The problem studied  in this paper concerns a two dimensional quantum system with a straight waveguide of width $d$ perturbed by a distant potential, represented by the Kato class measure.
The   waveguide will be imitated by the potential
\begin{equation}\label{eq-potentialalpha}
    V_\alpha = -\alpha  (x_2) \chi_\Sigma (x)\,, 
\end{equation}
where
\begin{equation}\label{eq-Sigma}
\Sigma := \{x= (x_1, x_2 )\,:\, x_1\in \mathbb R,\, x_2\in [0,d ]\}\,,
\end{equation}
and   $\chi_{\mathcal A}(\cdot)$  denotes the characteristic function of the set $\mathcal A \subset \mathbb R^2$, i.e. it is equal $1$ on $\mathcal A$ and $0$ on
 $\mathcal A \setminus \mathbb R^2$; moreover, $ \alpha \,:\, [0\,,d] \to \mathbb R$ stands for a piecewise continuous function. Following the notation introduced in~\cite{Exner21}, we will call the  interaction defined by $V_\alpha$ `soft waveguide' to emphasize the possibility of tunnelling out of $\Sigma$.
As we have already mentioned, apart from the  waveguide $\Sigma$, we introduce into the system a trap potential  localized on a  
compact set $\Omega \subset \mathbb R^2$
and defined by
a Radon measure $\mu_\Omega \equiv \mu  $  belonging to the Kato class and  supported on $\Omega$. We assume that $\Omega \cap \Sigma  = \emptyset$ and reserve  the special notation $\rho$ for  the distance $\rho:=\min_{x\in \Sigma , y\in \Omega } |x-y|\,$. 
The Hamiltonian  of the system can be defined as the operator  associated to  the sesquilinear form
\begin{equation}\label{eq-formalphabeta0}
     \mathcalligra { E}_{\alpha , \beta \mu} [f]= \int_{\mathbb R^2 } |\nabla f (x) |^2 \,\mathrm d x- \int_{\Sigma }\alpha(x_2)|f (x)|^2 \,\mathrm d x
- \beta  \int_{\mathbb R^2 }|I_\mu f (x)|^2 \,\mathrm d \mu (x)\,, \quad f\in W^{1,2 } (\R^2 )\,,
\end{equation}
where $W^{1,2}(\R^2 )$ stands for the Sobolev space
$W^{1,2}(\R^2 )= \{f\in L^{2}(\R^2 ): \partial_i f\in L^{2}(\R^2 )\,,\, i=1,2 \}$\footnote{The symbols $\partial_i$ stand for the distributional derivatives labelled by two cartesian variables if $\R^2$.  }
and
$I_\mu$ denotes the space embedding $I_\mu \,:\,  W^{1,2 } (\R^2 ) \to   L^2 _\mu \equiv L^2 (\mathbb R^2, \,\mathrm d \mu)$ and $\beta \in \mathbb R$.  The role of the parameter $\beta $ at this general stage is to  maintain consistency  with further notations.
The resulting Hamiltonian will be denoted  $H_{\alpha , \beta \mu}$.


The perturbation determined by $\mu$
allows various  realisations that are interesting from a quantum point of view.  For example, we can consider a  continuous function $V\,:\, \mathbb R^2\to \mathbb R$   with the support $\Omega$. Then, for any Borel set $\mathcal B \subset \mathbb R^2$, we have $\mu (\mathcal B) = \int_{\mathcal B} V(x) \mathrm d x$.
Note that $\Omega \subset \mathbb R^2$ can be defined as  a set of a  lower dimension.
This leads to another important example of the measure $\mu $ defined by the delta potential, named also the delta interaction. To explain in more detail, assume that $\Omega $ is a finite line (loop or open arc)  in $\mathbb R^2$.  Then   $\mu (\mathcal B) = \mathrm {lin} (\mathcal B \cap \Omega  ) $, where $\mathrm{lin} (\cdot)$ stands for the linear measure. In this situation the perturbation formally corresponds to  $-\beta \delta (x-\Omega  )$, where $ \delta (\cdot -\Omega  )$ stands for the Dirac delta supported on $\Omega $. The examples of  perturbations belonging  to the Radon measure class  are shown in Fig.~\ref{fig:drawing}. For the clarity of the presentation  we will first demonstrate the reasoning for the simplified model, namely $\mu (\mathcal B) = \int_{\mathcal B\cap \Omega } \mathrm d x $, where  $\Omega \subset \mathbb R^2$ is  compact and $\alpha (\cdot ) =\alpha $, with $\alpha >0$.  For this  case we use a special notation for the Hamiltonian:
\begin{equation} \label{eq-Hamalphabeta-gen1}
   H_{\alpha, \beta} =
   -\Delta - \alpha \chi_\Sigma (x)  - \beta \chi_\Omega  (x)\,.
\end{equation}
\begin{figure}
    \centering
    \includegraphics[width=0.7\textwidth]{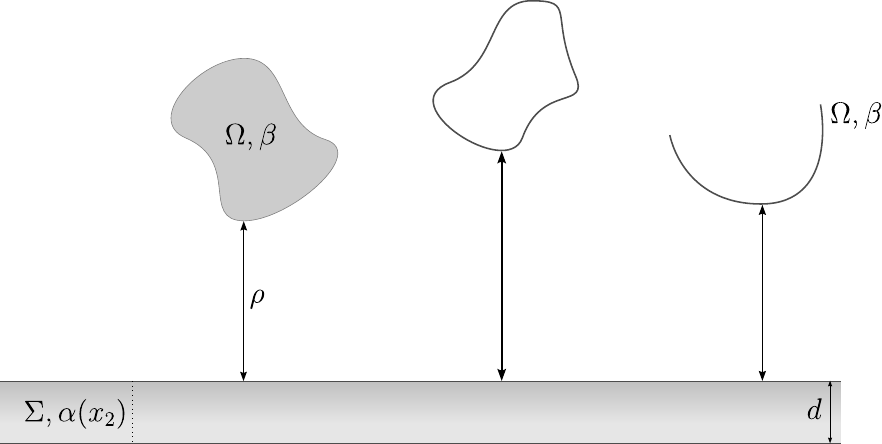}
    \caption{A schematic illustration of the model. $\Sigma$ represents the structure of the waveguide, and $\Omega$ represents the perturbation.}
    \label{fig:drawing}
\end{figure}
Modern nanotechnology allows the fabrication of semiconductors with different shapes and geometries. Leveraging a high mobility material such as $Ga_{1-x}Al_xAs$ we can produce a quasi one or two
dimensional electron gas in quantum wires, waveguides or layers. The question of what properties of semiconductor structures affect the quantum transport, is essential from  the perspective of the functioning of microelectronic devices. Inspired by the  question of the relation between architecture  of the semiconductor structure and spectral characteristics, we  investigate the resonant quantum model with an open channel and a trap inducing resonances due to the tunneling. Such a system is govened by a certain realization of $H_{\alpha, \beta \mu}$, for example by (\ref{eq-Hamalphabeta-gen1})  and our main aim is to understand how the distance between $\Sigma$ and $\Omega$, as a simple parameter of the system architecture, affects the resonances.

To take a closer look at  the spectral situation we auxiliary assume
that $\beta =0$ and consider
%
the system governed by $H_{\alpha, 0}$, i.e. which is employed only with the straight waveguide. The translational invariance allows to decompose $H_{\alpha, 0}$ onto longitudinal and transversal components. The latter  is determined by the one dimensional Hamiltionian $\mathrm h_\alpha $ with the potential $-\alpha (\cdot )\chi_{[0, d]}$. Obviously,  the essential spectrum of $\mathrm h_\alpha $ ranges from $0$ to infinity. In the following we assume that $\mathrm h_\alpha $ admits also discrete spectral point below $0$.
Let $\{E_{\alpha;1}\leq ...\leq E_{\alpha , N_\alpha }\}$  stand for  the complete set of the discrete eigenvalues of $\mathrm{h}_\alpha$; we denote also $\mathcal{N}_\alpha = \{1,2,..., N_\alpha \}$.
 Backing to the two  dimensional system with the  Hamiltonian $H_{\alpha, 0}$ we can state  that its spectrum covers
\begin{equation}\label{eq-essentialalphazero}
   \sigma (H_{\alpha, 0 }) =  \sigma_{\mathrm{ess}} (H_{\alpha, 0 }) = [E_\alpha \,,\infty ) \,,
\end{equation}
where $E_\alpha = E_{1; \alpha }$ is the lowest eigenvalue of $\mathrm h _\alpha $.
Under certain conditions on $\mu$, we can also show that the essential spectrum remains stable under the perturbation which can be formally written as  $-\beta \mu $, i.e. $\sigma _{\mathrm{ess}} (H_{\alpha,  \beta \mu}) = [E_\alpha \,,\infty ) $. However, it can produce discrete eigenvalues below the threshold of the essential spectrum.


The main topic of this paper it devoted to the problem of resonances that may emerge if the distance $\rho$ is very large. The resonances and the related metastable states
are a kind of spectral memory about the decoupled systems $H_{\alpha , 0}$ and  $H_{0,\beta \mu}$, which formally corresponds to  $\rho=\infty $.
 To explain in more detail, assume that $\alpha = 0$ and then the corresponding Hamiltonian $H_{0,\beta\mu }$ has the eigenvalues $\{\mathcal{E}_{\beta;1 } \leq ...\leq \mathcal{E}_{\beta;N_\beta }\} $, $\mathcal{N}_\beta := \{1,..,N_\beta \}$ with a corresponding eigenfunction $\{\omega_{\beta;n} \}_{n\in \mathcal N_\beta }$.
 If for a given $n$ there exists $\mathcal E_{\beta ; n}$ such that  $\mathcal{E}_{\beta;n} > E_\alpha$  then we may  expect that a metastable state with the resonant energy $\mathcal{E}_{\beta;n}$ can  appear in the system governed by $H_{\alpha, \beta \mu}$. This will be manifested as a second sheet resolvent pole  residing in a certain neighbourhood of $\mathcal{E}_{\beta;n}$.

\emph{Main results of the paper.} The main result of this paper can be formulated as follows.
\begin{itemize}
    \item Assume $\mathcal{E}_{\beta;n}$ is a simple eigenvalue of $H_{0, \beta \mu}$ and it belongs to
     $(E_{\alpha;j}, E_{\alpha;j+1})$ for $j\in \mathcal N_\alpha $. The resolvent of  $H_{\alpha, \beta \mu }$ has a second sheet analytic continuation to the lower half-plane and  there exists a unique resolvent pole in a neighbourhood of $\mathcal{E}_{\beta;n} $  that admits the asymptotics
 \begin{equation} 
   z_{n}(\rho )= \mathcal E_{\beta ; n}+\mathcal O \Big(\frac{\e^{-\sqrt{2 |\mathcal E_{\beta ;n}|} \rho}  }{\rho }\Big)\,.
\end{equation}
\item The dominated term of the imaginary component of pole takes the form
\begin{align} \label{eq-Gamman0-intro}
\begin{split}\Gamma _n (\rho )=  -\beta ^2 \sum_{k=1}^j\frac{\pi}{ 2\sqrt{ \mathcal E_ {\beta;n }- E_{\alpha; k}  }
} & \Big[ \Big| \int_{\Omega } I_\mu (\omega_ {\beta; n } (x)^{\ast } \varphi_{k} (x,p_1 ))\mathrm d \mu (x) \Big|^2 _{p_1=(\mathcal E_{\beta; n}- E_{\alpha; k } )^{1/2 }} \\ &+ \Big| \int_{\Omega } I_\mu (\omega_ {\beta; n } (x)^{\ast } \varphi_{k} (x,p_1 ))\mathrm d \mu (x) \Big|^2 _{p_1=-(\mathcal E_{\beta; n}- E_{\alpha; k } )^{1/2 }} \Big]\,,
\end{split}
\end{align}
where $\varphi_{k}(x,p_1)$ are the generalized eigenfunctions of $H_{\alpha,0}$ and let  us recall that $\omega_{\beta;n}$ stand for the eigenfunctions of $H_{0, \beta \mu}$; the above integrals depend on the argument $p_1$ and the indexes mean that they are determined at the specific points $p_1=\pm (\mathcal E_{\beta; n}- E_{\alpha; k } )^{1/2 }$.
For the special case of Hamiltonian defined by (\ref{eq-Hamalphabeta-gen1}),
the dominated term turns to
\begin{align} \label{eq-Gamman0}
\begin{split}
\Gamma _n (\rho )=  -\beta ^2 \sum_{k=1}^j\frac{\pi}{ 2\sqrt{ \mathcal E_ {\beta;n }- E_{\alpha; k}  }
} & \Big[ \Big| \int_{\Omega }\omega_ {\beta; n } (x)^{\ast } \varphi_{k} (x,p_1 )\mathrm d x \Big|^2 _{p_1=(\mathcal E_{\beta; n}- E_{\alpha; k } )^{1/2 }} \\ &+  \Big| \int_{\Omega }\omega_ {\beta; n } (x)^{\ast } \varphi_{k} (x,p_1 )\mathrm d x \Big|^2 _{p_1=-(\mathcal E_{\beta; n}- E_{\alpha; k } )^{1/2 }} \Big]\,.
\end{split}
\end{align}
\end{itemize}
 From the physical point of view, the imaginary component  of pole plays the main rule in the time evolution of metastable state.
Fermi's golden rule plays an essential role in physics.
 It describes the rate of transition from an energy eigenstate of a quantum system to an energy continuum as an effect of a weak perturbation.
The corresponding formula was, first time, obtained by Dirac \cite{Dirac} and then, after more than twenty years, it was analyzed and called as `golden rule' by Fermi \cite{Fermi}. From the physical perspective, distantly located trap allows for controlling scattering processes and quantum transport in the waveguide, influencing tunneling and dissipative processes. The distance between the trap and the waveguide can be adjusted to impact the system dynamics.
Studying these decay mechanisms, especially through explicit derivations like Fermi's golden rule, is crucial for optimizing such systems. This provides insights into key components of quantum technology
such as wires and waveguides,  and facilitates the optimization of related processes.

We would like to conclude this introduction by saying a few words about how the results obtained in this paper fit into a wider context. The work belongs to the line of research that is dedicated  to the  exploration of the
spectral properties of quantum systems with soft waveguides. These types of systems
have been analyzed mainly from the point of view of the existence of bound states, cf.~\cite{ Exner21, Exner22, EKL23, ExnerS2023, ExnerV, Teufel2, KK22}. The problem of resonances in soft waveguides remains a lot of questions open.  It is worth mentioning that  the problem of resonances induced by  breaking the symmetry in the  soft waveguide systems has been analyzed in \cite{Kondej2024} by
the first author of the present paper.
In \cite{Kondej2024} the resonances were memories of the embedded eigenvalues
preserved by symmetry of the system. In the current paper, the resonances are induced  by the tunneling, which means that the strategy for recovering resonance pole  is necessarily not the same, although the main technical concepts are borrowed from \cite{Kondej2024}.
On the other hand, in the present paper we discuss systems
with  a general class of tunnel traps  defined by the Kato class measures.
This aspect is interesting from both mathematical and physical perspectives. First, it demonstrates that the resonance techniques (partially developed in \cite{Kondej2024}) can be extended to a much broader scenario. Moreover, it shows that the trap can be defined not only by a regular potential but also by lower-dimensional structures, such as wires. The latter is interesting, in particular, from the physical point of view.
Additionally, the present paper provides an explicit formulation of Fermi's golden rule for this generalized case, which in \cite{Kondej2024} was presented only implicitly in terms of the spectral resolution, a formulation less familiar to physicists.

Note that  the problem of resonances induced by the effect of tunnelling in the model with an infinite straight line and a distant point interaction has been analyzed in
\cite{ExnerKondej04}
and, on the other hand,  the resonances in the three dimensional systems with point potentials have been investigated, for example in~\cite{CAZEG,LipLot}, see also~\cite{AGHH}.
It is also worth mentioning that problem of resonances
in, so called, hard wall (Dirichlet) waveguides   has been undertaken in series of papers,  cf.~\cite{BG, DeGr, DEM, KovarikS, KSJC}  see also references in~\cite{EK2015}. However, in these systems, contrary to  the model discussed in this paper, a particle cannot escape from the waveguide because of the infinite potential barrier at the boundary.

The paper is organized as follows: In Section~\ref{sec-description} we provide the description of Hamiltonians $H_{\alpha , 0}$, $H_{\alpha , \beta }$ and  $H_{\alpha , \beta }$  corresponding to~(\ref{eq-Hamalphabeta-gen1}). We state preliminary facts about the spectra and derive the resolvents of the investigated Hamiltonians. In Section~\ref{sec-resonances} we show that, under certain conditions,  the resolvent of $H_{\alpha, \beta }$ has
a unique second sheet pole and derive its asymptotics. In Section~\ref{sec-Fermi} we analyze the imaginary part of the pole and show that Fermi's golden rule is fulfilled in the
form of~(\ref{eq-Gamman0}). Finally, in Section~\ref{sec-Kato} we generalize the result for large class of potential defined by $H_{\alpha, \beta \mu }$ and related to (\ref{eq-formalphabeta0}).

\section{Description of simplified  model: influence of distant well on a straight waveguide} \label{sec-description}
In this section, we consider a two dimensional quantum mechanical model with the potential supported on a straight, infinite planar waveguide and a distant well potential located on a compact set $\Omega \subset \mathbb{R}^2$ with the smooth boundaries $\partial \Omega$. We assume that the waveguide has a width $d$, therefore it can be defined as follows
\begin{equation}\label{eq-Sigma}
\Sigma := \{x= (x_1, x_2 )\,:\, x_1\in \mathbb R,\, x_2\in [0,d ]\}\,.
\end{equation}
Furthermore, we assume that $\Omega \cap \Sigma =\emptyset $ and denote the distance
\begin{equation}\label{eq-defrho}
\rho:=\min_{x\in \Sigma , y\in \Omega } |x-y|\,.
\end{equation}
Due to the orientation of $\Sigma $ we can also write $\rho=\min_{x\in \Sigma , y\in \Omega } |x_2-y_2|$.
The parameter $\rho$ will play the essential role in the further discussion. 

Let $\alpha$ and $\beta$ be positive numbers, and using them, we define the Hamiltonian
\begin{equation}\label{eq-Hamiltonianalphabeta}
    H_{\alpha, \beta }= -\frac{d^2}{dx_1^2}-\frac{d^2}{dx_2^2} - \alpha \chi_\Sigma (x)  -\beta \chi_\Omega (x)\,:\, W^{2,2} (\mathbb R^2)\to L^2 (\mathbb R^2)\,,
\end{equation}
where the Sobolev space $W^{2,2} (\mathbb R^2)=  \{f\in L^{2}(\R^2 ): \partial_i \partial_j f\in L^{2}(\R^2 )\,,\, i, j\in\{1,2\} \}$ determines the maximal domain. In fact, the above Hamiltonian depends on $\Sigma$ and $\Omega$ as well, however, to avoid over-indexing, we do not involve $\Sigma$ and $\Omega$ into the notation for the Hamiltonian. $H_{\alpha, \beta }$ defined by (\ref{eq-Hamiltonianalphabeta}) determines a particular example of (\ref{eq-potentialalpha}).

\subsection{Insight into a  strip model}

We auxiliary consider the situation  $\beta =0$, i.e. we remove the well potential supported on $\Omega$, leaving only the strip interaction. The corresponding Hamiltonian $H_{\alpha, 0}$ has the translational symmetry, so we can decouple it  into  transverse and longitudinal components. The latter  carries    the information only about the kinetic energy. The transverse component of Hamiltonian $H_{\alpha, 0}$ is defined by
\begin{equation} \label{eq-trasvHam}
    \mathrm h_\alpha = -\frac{d^2}{dx_2^2} - \alpha \chi_{[0, d]} (x_2) \,:\, W^{2,2} (\mathbb R)\to L^2 (\mathbb R)\,,
\end{equation}
which determines the standard one dimensional Hamiltonian with a rectangular well of depth $\alpha $. The negative eigenvalues and eigenfunctions of
$\mathrm h _\alpha $ formed in the ascending order will be denoted
\begin{equation}\label{eq-evaluesalpha}
\{E_{\alpha ; 1}\leq E_{\alpha ; 2}\leq ...\leq E_{\alpha ; N_\alpha} \}\,,
\end{equation}
and
\begin{equation}\label{eq-N_alpha}\{\phi_{\alpha ; n}\}_{n\in {\mathcal N_\alpha }}
\,,\,\,\,\,\,\,\, \mathcal N_\alpha := \{1,...N_\alpha \}\,
\end{equation}
stand for the corresponding eigenfunctions.
Then  $H_{\alpha, 0}$  can be represented as
\begin{equation} \label{eq-trasldecom}
    H_{\alpha, 0} = -\frac{d^2}{dx_1^2} \otimes \mathrm I +\mathrm I \otimes \mathrm{h}_\alpha \,.
\end{equation}
Consequently, the spectrum of $H_{\alpha, 0 } $ admits only the essential component of the form
\begin{equation}\label{eq-essentialalpha}
    \sigma_{\mathrm{ess}} (H_{\alpha, 0 }) = [E_\alpha \,,\infty ) \,, \quad \mathrm{where} \quad E_\alpha = E_{\alpha ; 1}\,.
\end{equation}

\noindent \emph{Notation.} In the following discussion we will use the notation $\mathrm I$ for the unit operator independently of space where it acts.
\subsubsection{Resolvent of $H_{\alpha, 0}$}
In this section  we are going to derive the resolvent of $H_{\alpha, 0}$ which will be the basic tool in the analysis of resonances.
Let us recall that the 'free' resolvent is given by \( R(z) = (-\Delta - z)^{-1} \), and it can be represented as an integral operator with a kernel defined by the Macdonald function \( K_0(\cdot) \), also known as the modified Bessel function of the second kind.\footnote{Alternatively, the kernel of \( R(z) \) can be expressed using the Hankel function as \( \frac{i}{4} H_0(\sqrt{z}|\cdot - \cdot|) \). Specifically:
\begin{equation} \label{eq-Macdonald}
R(z) f = \frac{1}{2\pi} \int_{\mathbb{R}^2} K_0(k_z |\cdot - y|) f(y) \, \mathrm{d}y, \quad f \in L^2(\mathbb{R}^2),
\end{equation}
where \( k_z = -i \sqrt{z} \) with \( \Im \sqrt{z} > 0 \) and the kernel of \( R(z) \) defines the Green's function of the operator \( -\Delta - z \).
}

In the present case, when the potential is defined by a negative coupling constant $-\alpha $, the Birman--Schwinger operator  is  given by $
-\alpha \chi_{\Sigma }R(z)\chi_{\Sigma }\,:\, L^2 (\Sigma )\to
L^2 (\Sigma )\,.$ Consequently, for $z$ living in the   resolvent set $\varrho (H_{\alpha, 0})$ the  operator
\begin{equation}\label{BS_kernel}
\mathrm I -\alpha \chi_{\Sigma }R(z)\chi_{\Sigma }\,:\, L^2 (\Sigma )\to
L^2 (\Sigma )
\end{equation}
is invertible.
Using the Birman--Schwinger concept we represent the resolvent  $R_{\alpha, 0 }(z) = (H_{\alpha, 0 }-z)^{-1}$ as
\begin{equation} \label{eq-BSralpha}
    R_{\alpha, 0 }(z)= R(z) +\alpha R(z) \chi_{\Sigma }[\mathrm I -\alpha \chi_{\Sigma }R(z)\chi_{\Sigma }]^{-1}\chi_{\Sigma } R(z)\,,
\end{equation}
cf.~\cite{BEG, BEKS, HamsmanKrejcirik, Po1, Simon-comprehensive}.

To proceed further we are going to  present the resolvent $R_{\alpha, 0 }(z)$ in an alternative  form that reveals transverse symmetry and reflexes  the decomposition (\ref{eq-trasldecom}); a similar construction has been carried out in \cite{EKL23, Kondej2024}. Note that the spectral family (equivalently, the resolution of the identity) associated to the operator  $\mathrm h_\alpha $ is determined the  eigenprojectors
$P_{\alpha; n } = (\phi_{\alpha;n }, \cdot )_{L^2 (\mathbb R)}\phi_{\alpha;n }$ and the component corresponding to the continuum of $\mathrm h_\alpha $ which we denote  $E_\alpha (\cdot ).$
Note that the longitudinal component of \( H_{\alpha, 0} \), defined by \( -\frac{d^2}{dx_1^2} \otimes \mathrm{I} \) on \( L^2(\mathbb{R}) \otimes L^2(\mathbb{R}) \), reproduces the spectral family for the operator \( -\frac{d^2}{dx_1^2} \) acting in \( L^2(\mathbb{R}) \). We denote it as \( E_0(\cdot) \), and the explicit form of \( E_0(\cdot) \) is given by:
\begin{equation}\label{eq-spectal-family}
E_0(\lambda) f = \frac{1}{(2\pi)} \int_{|\xi|^2 \leq \lambda} f(\xi) e^{i x \cdot \xi} \, \mathrm d\xi.
\end{equation}
By means of the introduced notation we can write
\begin{equation}\label{eq-Ham-res}
H_{\alpha, 0} \, (\mathrm I \otimes P_{\alpha ;n} )=  \int_{[0, \infty )}(\lambda + E_{\alpha ;n} )\,\mathrm d E_0 (\lambda ) \otimes P_{\alpha ;n}\,.
\end{equation}
This, consequently, allows to define components
of $R_{\alpha, 0} (z)$
related to the projection onto
$L^2 (\mathbb R)\otimes P_\alpha L^2 (\mathbb R)$, where  $P_{\alpha }= \sum_{n\in \mathcal N_\alpha } P_{\alpha, n}$, and
$L^2 (\mathbb R)\otimes P_\alpha^\bot L ^2 (\mathbb R)$, where $P_\alpha^\bot$ stands for the orthogonal complement of $P_\alpha$.
Precisely,
\begin{equation} \label{eq-Rd0}
R^d _{\alpha, 0}(z) := R _{\alpha, 0}(z) \,  (\,\mathrm I \otimes P_\alpha\,) =\sum_{n\in \mathcal N_\alpha }\int_{[0, \infty )} \frac{\mathrm d E _0(\lambda )}{\lambda + E_{\alpha; n } -z }\otimes P_{\alpha; n }\,
  \end{equation}
and
\begin{equation} \label{eq-Rc}
R^c _{\alpha, 0}(z) := R _{\alpha, 0}(z) \, (\,\mathrm I \otimes P_\alpha^\bot \,= \int_{[0, \infty )} \int_{[0, \infty )} \frac{\mathrm d E_0(\lambda ) \otimes
\mathrm d  E_\alpha  (\lambda ')
}{\lambda +\lambda '-z }\,.
\end{equation}
Note that   $E_0 (\cdot )$ can be also expanded by means of the plane waves $(2\pi)^{-1/2}\e^{ix_1p_1}$ which yields
\begin{equation} \label{eq-Four}
\Big(-\frac{d^2}{dx_1^2} + E_{\alpha;n} -z \Big)^{-1} f (x_1)=\int_{[0, \infty )} \frac{\mathrm d  E_0  (\lambda )}{\lambda + E_{\alpha; n } -z } f(x_1)=
\frac{1}{2\pi }\int_\mathbb R  \, \int_\mathbb R
\frac{\e^{i p_1
(x_1 -y_1)  }}{p_1^2 + E_{\alpha; n } -z  }
f(y_1)
\,\mathrm d p_1 \,\mathrm d y_1 \,.\end{equation}
Furthermore, the one dimensional Green's function which appears in the last equivalence admits also representation
\begin{equation} \label{eq-Green1D}\frac{1}{2\pi }\int_\mathbb R  \,
\frac{\e^{i p_1
(x_1 -y_1)  }}{p_1^2 + E_{\alpha; n } -z  }
\,\mathrm d p_1= \frac{i}{2} \frac{\e^{i\tau_{\alpha ; n} (z)
|x_1 -y_1|  }}{\tau_{\alpha ; n} (z) } \,,
\end{equation}
where $
\tau_{\alpha ; n} (z):= \sqrt {z - E_{\alpha ;n }}$  with the square root determined on the first analytical sheet, i.e. $\Im \tau_{\alpha ; n} (z)>0$.
Summarizing, we get the decomposition
\begin{equation} \label{eq-Rdc}
    R_{\alpha, 0 }(z)=R_{\alpha, 0 }^d(z)+R_{\alpha, 0 }^c(z)\,,
\end{equation}
where 
\begin{align}\begin{split}\label{eq-Rd}
R_{\alpha, 0 }^d(z) f(x)= &
\frac{1}{2\pi}\sum_{n\in \mathcal N_\alpha }\int_{\mathbb R^2}  \int_\mathbb R
\frac{\e^{i p_1
(x_1 -y_1)  }}{p_1^2 + E_{\alpha; n } -z  }
\phi_{\alpha ; n } (x_2) \phi_{\alpha ; n } (y_2) ^\ast f(y)\,\mathrm d p_1 \,\mathrm d y
\\ = &
\frac{i}{2}\sum_{n\in \mathcal N_\alpha }\int_{\mathbb R^2} \frac{\e^{i\tau_{\alpha ; n} (z)
|x_1 -y_1|  }}{\tau_{\alpha ; n} (z) }
\phi_{\alpha ; n } (x_2) \phi_{\alpha ; n } (y_2)^\ast  f(y) \,\mathrm d y \,,
\end{split}
\end{align}
where $y=(y_1, y_2)$
and 
\begin{equation}\label{eq-Rc}
R_{\alpha, 0 }^c(z) f(x)= \frac{i}{2}\int_\mathbb R \int_\mathbb R \frac{\e^{i\tau_{\alpha } (z, p_2)
|x_1 -y_1|  }}{\tau_{\alpha } (z, p_2) }
\psi_{\alpha  } (x _2, p_2) \psi_{\alpha } (y_2, p_2) f(y) \,\mathrm d y\,\mathrm d p_2\,,
 \end{equation}
 where $\tau_{\alpha } (z, p_2) :=  \sqrt {z  - p_2^2}$ with $\Im \sqrt {z  - p_2^2}>0$, and $\psi_{\alpha } (\cdot , \cdot)$ stand for the generalized eigenfunctions of $\mathrm{h_\alpha}$ (constructed from plane waves). The similar decomposition were used in \cite{EKL23}.

\subsection{Quantum well system without straight strip}
In order to collect the most important information about the spectral properties of the system governed by $H_{\alpha, \beta }$, in this section we will auxiliary consider $H_{0, \beta }$.
Since
the potential is negative we can conclude that  the negative discrete spectrum
is non empty, cf.~\cite[XIII.11]{RS4}. The  discrete eigenvalues of $H_{0, \beta }$ formed in the ascending order will be denoted as
\begin{equation}\label{eq-eigenvalubeta}
\{\mathcal E_{\beta ; 1}\leq  \mathcal E_{\beta ; 2}\leq ...
\leq \mathcal E_{\beta ; N_\beta }\}
\end{equation}
and, let
\begin{equation}\label{eq-eigenvalubeta1}
\{\omega_{\beta  ; n}\}_{n\in {\mathcal N_\beta  }}\,,\,\,\,\,\,\,\, \mathcal N_\beta := \{1,...N_\beta  \}\,,
\end{equation}
stand for the corresponding eigenfunctions.

We again use the Birman--Schwinger resolvent $R_{0, \beta}(z) = (H_{ 0 , \beta }-z)^{-1}$ which in this case takes the form
\begin{equation} \label{eq-betaresolvent1}
    R_{0,\beta  }(z)= R (z) +\beta  R(z) \chi_{\Omega  }(\mathrm I -\beta  \chi_{\Omega  }R(z)\chi_{\Omega  })^{-1}\chi_{\Omega  } R(z)\,, \quad  z\in  \varrho (H_{0,\beta }) \,.
\end{equation} 
The Birman--Schwinger principle allows to associate zeros of $\mathrm I -\beta  \chi_{\Omega  }R(\mathcal E_{\beta ;n})\chi_{\Omega  }\,:\,L^2 (\Omega) \to L^2 (\Omega )$ and eigenfunctions $\omega_{\beta; n }$ of $H_{0,\beta }$ in the following way
\begin{equation}\label{eq-birmanschwingeromega}
    \omega_{\beta; n} = R(\mathcal E_{\beta; n })\chi_{\Omega }w_{\beta; n}\,, \quad \mathrm{where} \quad w_{\beta; n} \in \ker [\mathrm I -\beta  \chi_{\Omega  }R(\mathcal E_{\beta ; n})\chi_{\Omega  }]\,  \quad \text{and} \quad n\in \mathcal N_\beta \,,
\end{equation}
cf.~\cite{BEG, BEKS, HamsmanKrejcirik, Po1}. 
Let $\mathcal S_n$ denote a sufficiently small neighbourhood of  $\mathcal E_{\beta; n}$. With a view to further discussion we prove the following statement.

\begin{lemma} Assume that $\mathcal E_{\beta ; n}$ is a simple discrete eigenvalue of $H_{0,\beta }$ and  $z\in \mathcal S_n\setminus \{ \mathcal E_{\beta; n}\}$. Then
   the  operator $(\mathrm I -\beta  \chi_{\Omega  }R(z)\chi_{\Omega  })^{-1}\,:\,L^2 (\Omega )\to L^2 (\Omega )$ admits the representation
    \begin{equation}\label{eq-birmanschwingerpoles}
(\mathrm I -\beta  \chi_{\Omega  }R(z)\chi_{\Omega  })^{-1} = \frac{1}{\beta } \frac{1}{\mathcal E_{\beta ; n } - z}\mathrm P_{\beta; n }+A_n(z)\,,
\end{equation}
where $\mathrm P_{\beta; n }= (w_{\beta; n}\,, \cdot )_{L^2 (\Omega )}w_{\beta; n}$ and $A_n(z)$ is  analytic in  $\mathcal S_n$.
\end{lemma}
\begin{proof}
Note that  $R_{0, \beta } (z)$ is the resolvent of a self adjoint operator
with isolated spectral point at $\mathcal E_{\beta; n }$ with corresponding eigenvector $\omega_{\beta; n}$ and eigenprojector  $P_{\beta;n }$.
Thus, for $z\in \mathcal S_n\setminus \{ \mathcal E_{\beta; n}\}$ we can write
\begin{equation} \label{eq-betaresolvent2}
R_{0, \beta } (z)= \frac{1}{\mathcal E_{\beta ; n } - z}P_{\beta ;n }+B_n (z)\,:\,L^2 (\mathbb R^2)\to L^2 (\mathbb R^2)\,,
\end{equation}
where   $B_n(z)$ acts in $L^2 (\mathbb R^2 )$ and it is analytic in $\mathcal S_n$; the above follows from the spectral  resolution of the $R_{0, \beta  }(z)$, cf.~\cite[V-(3.21)]{Kato}. The singularities coming from the remaining isolated points contribute to $B_n(z)$.
On the other hand, since $(\mathrm I -\beta  \chi_{\Omega  }R(z)\chi_{\Omega  })^{-1}$ has the same pole of the multiplicity one, it can be decomposed as
\begin{equation}\label{eq-analyticBS}
(\mathrm I -\beta  \chi_{\Omega  }R(z)\chi_{\Omega  })^{-1}= \frac{1}{\mathcal E_{\beta;n }-z } \tilde P_n (z)+\tilde B_n(z)\,:\,L^2 (\Omega )\to L^2 (\Omega )\,,
\end{equation}
where $\tilde B_n(z)$ stands for an analytic component and $\tilde P_n (z)$ is a projector acting in the space $L^2 (\Omega)$, cf.~\cite{GHN2015}.  Comparing the above with (\ref{eq-betaresolvent2}) and taking into account the resolvent formula (\ref{eq-betaresolvent1}),
we can write
\begin{equation} \label{eq-Pbeta}
 P_{\beta ;n }= \beta R(z) \chi_\Omega \tilde P_n(z) \chi_\Omega  R(z) \,.
 \end{equation}
Furthermore, using   the analyticity of the resolvent $R(z)$ for   $z\in\mathcal S_{ n }$, i.e.
\begin{equation} \label{eq-R(z)}
R(z) = R(\mathcal E_{\beta ; n} ) + (\mathcal E_{\beta ; n}-z ) R(z) R(\mathcal E_{\beta ; n})\,
\end{equation}
we obtain from (\ref{eq-betaresolvent1}) that $  \beta R( \mathcal E_{\beta ; n} ) \chi_\Omega \tilde P_n(  \mathcal E_{\beta ; n}) \chi_\Omega  R( \mathcal E_{\beta ; n}) = P_{\beta ;n }+ \mathcal O_{\mathrm {op.u.}}(z-\mathcal E_{\beta ; n})$, where the error term is understood in the norm operator sense.
Implementing $\mathrm P_{\beta; n }= (w_{\beta; n}\,, \cdot )_{L^2 (\Omega )}w_{\beta; n}$ which, by definition of $\omega _{\beta ; n}$ and $w _{\beta , n}$ (see (\ref{eq-birmanschwingeromega})) satisfies
$P_{\beta ; n } =R (\mathcal E_{\beta ; n } ) \mathrm P_{\beta ; n }R (\mathcal E_{\beta ; n } )$, we state that $\tilde P_{n} (\mathcal E_{\beta ; n }) = \beta^{-1}\mathrm P_{\beta;n }$ and
eq.~(\ref{eq-analyticBS}) turns to
\begin{equation}
(\mathrm I -\beta  \chi_{\Omega  }R(z)\chi_{\Omega  })^{-1} = \frac{1}{\beta } \frac{1}{\mathcal E_{\beta ; n } - z}\mathrm P_{\beta; n }+A_n(z)\,,
\end{equation}
where  $A_n(z)$ acting  in $L^2 (\Omega )$ is  analytic in  $ \mathcal S_n$. This completes the proof.
\end{proof}

\subsection{Resolvent of $H_{\alpha, \beta }$ and preliminary facts about its spectral properties}

Considering as a starting point the system governed by $H_{\alpha , 0}$ with corresponding resolvent $R_{\alpha, 0}(z)$, given by (\ref{eq-BSralpha}), we implement the pertubation determined by $-\beta\chi_\Omega\,$. To apply this strategy we assume that a complex number $z$ is such that the operator $\mathrm I -\beta  \chi_{\Omega  }R_{\alpha , 0 }(z)\chi_{\Omega  }\,:\, L^2 (\Omega )\to L^2 (\Omega ) $ is invertible.
After adding the potential $-\beta \chi_\Omega $ to the system governed by $H_{\alpha, 0}$ the resolvent takes the form
\begin{equation}\label{eq-resolvent-alpha-beta}
    R_{\alpha ,\beta  }(z)= R_{\alpha , 0} (z) +\beta  R_{\alpha, 0}(z) \chi_{\Omega  }(\mathrm I -\beta  \chi_{\Omega  }R_{\alpha ,0 }(z)\chi_{\Omega  })^{-1}\chi_{\Omega  } R_{\alpha, 0 }(z)\,,
\end{equation}
cf.~\cite{BEG, BEKS, HamsmanKrejcirik, Po1}.
\bigskip

\noindent \emph{Essential spectrum. }  The stability of the essential spectrum follows from Weyl's Theorem, cf.~\cite[Thm.~XIII.14]{RS4}.
Indeed, assuming the $z\in \varrho (H_{\alpha, 0})$ and employing compactness of $\Omega $ one can show that both   $R (z)\chi_\Omega \,:\,L^2 (\Omega )\to L^ 2(\mathbb R^2)$ and $\chi_\Sigma R (z)\chi_\Omega \,:\, L ^2 (\Omega )\to L^ 2(\Sigma )$ are compact; we postpone a detailed analysis to
the next section, see~Lemma~\ref{le-normsbounds}.
Applying now
 (\ref{eq-BSralpha})
we obtain the compactness  of $R_{\alpha , 0} (z)\chi_\Omega$.
This, in view of Weyl's Theorem, yields

\begin{equation}
\sigma_{\mathrm{ess}} (H_{\alpha, \beta  })= \sigma  _{\mathrm{ess}} (H_{\alpha, 0 }) = [E_\alpha \,,\infty ) \,.
\end{equation}

\begin{remark}
\rm{Although  the discrete spectrum of $H_{\alpha, \beta }$ does not belong to the main line of this paper, it is worth mentioning that the attractive potential $- \beta \chi_\Omega $ leads to at least one discrete eigenvalue below $E_\alpha $. One can establish this
by using the variational method and  constructing  a trial function with separate variables, where $\phi_\alpha $ represents the transverse direction and a positive propagating function is defined in the logitudinal direction. The details are presented in~\cite{Kondej2024}.
}
\end{remark} 

\section{Resonances induced by the tunnelling} \label{sec-resonances}

\subsection{Analytic extension of $\chi_\Omega R_{\alpha, 0 }(z)\chi_\Omega$}
\label{sec-analytic-extension}
To analyze the resonances we will adopt   the strategy of finding     the poles of the resolvent~(\ref{eq-resolvent-alpha-beta}) as  zeros of  $\mathrm I -\beta  \chi_{\Omega  }R_{\alpha ,0 }(z)\chi_{\Omega  }$. Therefore, our first aim is to construct a second sheet continuation of  $\chi_\Omega R_{\alpha, 0 }(z)\chi_\Omega$. We start our discussion from the analysis of $\tau_{\alpha ; n}(z)$ (contributing  $\chi_\Omega R_{\alpha, 0 }^d(z) \chi_\Omega$) which has been defined so far for the first Riemann sheet, i.e. $\Im \tau_{\alpha ; n}(z) >0$.
For each $n\in \mathcal N_\alpha $, the function $z\mapsto \tau_{\alpha; n} (z) $ admits a second sheet analytic continuation   through $[E_{\alpha ; n}, \infty)$ and the values of $\Im \tau_{\alpha ; n}(z) $
are negative in the lower complex half-plane. Consequently, 
for each $k\in \mathcal N_\alpha $,
we can construct the lower half-plane analytic continuation  of
\begin{equation}\label{eq-Rda}
\chi_\Omega R_{\alpha, 0 }^d(z) \chi_\Omega= \chi_\Omega(x) \Big[\frac{i}{2}\sum_{n\in \mathcal N_\alpha }\frac{\e^{i\tau_{\alpha ; n} (z)
|x_1 -y_1|  }}{\tau_{\alpha ; n} (z) }
\phi_{\alpha ; n } (x_2) \phi_{\alpha ; n} (y_2)^\ast \Big] \chi_\Omega(y)\,,\
 \end{equation}
via the segments $(E_{\alpha; k }, E_{\alpha; k +1})$.
More precisely, assuming that $z=\lambda +i\varepsilon$, where $\lambda \in ( E_{\alpha; k }, E_{\alpha; k +1})$ we move to the lower half-plane, i.e. varying with $\varepsilon$ from positive to negative values and defining  $\Im \,\tau_{\alpha;n }(z) <0$ if $n\leq k$ and $\Im \,\tau_{\alpha;n } (z) >0$ if $n>k$. We insert it into (\ref{eq-Rda}) and, in this way, we construct the second sheet component of $\chi _\Omega R^d_{\alpha , 0} (z)\chi_\Omega  \,:\,L^2 (\Omega )\to L^2 (\Omega )$.  On the other hand, the component $R^c _{\alpha , 0}(z)$ does not require any additional analysis, since it is analytic for $z=\lambda +i\varepsilon$, where $\lambda \in ( E_{\alpha; k }, E_{\alpha; k +1})$ and $\varepsilon$ belongs to a neighbourhood of $0$.
In view of (\ref{eq-Rdc}), this constitutes, for any $k\in \mathcal N_\alpha $, the second sheet continuation of $\chi_\Omega R_{\alpha, 0} (z)\chi_\Omega $    via $( E_{\alpha; k }, E_{\alpha; k +1})$ to a certain set $\mathcal C_{ k}\subset \mathbb C_-$. We remain  the same
notation for the second sheet operator valued function $\mathcal   C_{ k}\ni z \mapsto \chi_\Omega R _{\alpha, 0} (z) \chi_\Omega $. The above reasoning has also been adopted in~\cite{Kondej2024}.

\medskip


\subsection{Zeros of Birman--Schwinger operator}

To find out the poles of   resolvent determined by zeros of $\mathrm I - \beta \chi_\Omega R_{\alpha, 0 } (z) \chi_\Omega\,:\, L^2 (\Omega ) \to L^2 (\Omega)$ we
 adopt (\ref{eq-BSralpha}) which allows to state  that   $z$ satisfies  $\ker [\mathrm I - \beta \chi_\Omega R_{\alpha, 0 } (z) \chi_\Omega ] \neq \emptyset$
 if and only if 
\begin{equation} \label{eq-BSralpha-ext}
 \ker [ \mathrm I -\beta \chi_\Omega R (z)\chi_\Omega -G_{\alpha, \beta } (z) ]\neq \emptyset\,,
\end{equation}
where
\begin{equation}\label{eq-defG}
G_{\alpha, \beta } (z) = \beta \,\alpha \,\chi_\Omega R(z) \chi_{\Sigma }(\mathrm I -\alpha \chi_{\Sigma }R(z)\chi_{\Sigma })^{-1}\chi_{\Sigma } R(z) \chi_\Omega\,:\,L^2 (\Omega ) \to L^2 (\Omega )\,.
\end{equation}
 Our aim is to express the condition   (\ref{eq-BSralpha-ext}) in the terms of a scalar equation. Note that  $\beta \chi_\Omega R(z)\chi_\Omega$ determines the Birman--Schwinger operator for $H_{0, \beta }$. For $z$ living in a certain neighbourhood of the discrete eigenvalue $\mathcal E_{\beta; n}$ but $z\neq \mathcal E_{\beta; n}$ the operator $\mathrm I- \beta \chi_\Omega R(z)\chi_\Omega$ is invertible and its inverse has the representation (\ref{eq-birmanschwingerpoles}).
%
Therefore one gets
\begin{equation}\label{eq-decompkernel}
\mathrm I -\beta \chi_\Omega R (z)\chi_\Omega  - G_{\alpha, \beta } (z)= (\mathrm I- \beta \chi_\Omega R (z)\chi_\Omega)\Big[ \mathrm I -(\mathrm I- \beta \chi_\Omega R (z)\chi_\Omega)^{-1} G_{\alpha, \beta } (z) \Big]\,.
\end{equation}
To proceed further we estimate  the norms involved in the above formula.
\begin{lemma} \label{le-normsbounds}The operator $\chi_\Omega R (z) \chi_\Sigma \,:\, L^2 (\Sigma ) \to L^2 (\Omega )$ has the Hilbert--Schmidt norm
\begin{equation}\label{eq-asymR0}
\| \chi_\Omega R (z) \chi_\Sigma \|^2_{\mathrm{HS}} \leq C\frac{|\Omega |\ }{|z|^{1/2}\mathcal T_z}\, \frac{\e ^{-\sqrt{2}\mathcal T_z\rho } }{\rho}\,,
\end{equation}
where $\mathcal T_z:= \Im \sqrt{z} >0 $, $\rho$ is defined by (\ref{eq-defrho}) and $C$ is a positive constant.
\end{lemma}
\begin{proof} To prove the claim we will employ  the representation (\ref{eq-Macdonald}) which implies
\begin{equation} \label{eq-hilb}
\| \chi_\Omega R (z) \chi_\Sigma \|^2_{\mathrm{HS}} =\frac{1}{(2\pi )^2} \int_{\Omega } \int_{\Sigma }\big|K_0 (k_z|x-y|)\big|^2\,\mathrm{d}x\,\mathrm{d}y\,,
\end{equation}
where $k_z= -i\sqrt{z}$;  let us remind that $x=(x_1, x_2)$ and the same notation is adopted for  $y$.
Developing the above norm, we obtain
\begin{align}
\begin{split} \label{eq-hilb2}
\| \chi_\Omega R (z) \chi_\Sigma \|^2_{\mathrm{HS}}& =\frac{1}{(2\pi )^2} \int_{\Omega }\int_{\Sigma }\big|K_0 (k_z(\, (x_1-y_1)^2+(x_2-y_2)^2)^{1/2}\big|^2\,\mathrm{d}x\,\mathrm{d}y \\
 &=
\frac{d}{(2\pi )^2} \int_{\Omega }\int_{\mathbb R } \big|K_0 (k_z(\, (x_1 -y_1)^2+y_2^2)^{1/2}\big|^2 \, \mathrm{d}x_1\,\mathrm{d}y\,.
\end{split}
\end{align}
Using the asymptotics  of the Macdonald function for the large argument, \cite[ p.~378,~9.7.2]{AS},
  we obtain
\begin{equation}\label{eq-K00}
|K_0 (k_z \,\xi )|^2
\leq C '\frac{1  }{\xi |z|^{1/2} } \mathrm e^{-2  \mathcal T _z\xi  }\,,
\end{equation} for $\xi >\xi_0 $, where $C'$ and $\xi_0$ are  positive constants and $C'$ depends on $\xi_0$. Combining  the above inequality
and (\ref{eq-hilb2}) and employing
(\ref{eq-defrho})
together with the monotonicity   of the  function $\mathbb R_+ \ni x \mapsto x^{-1} \e^{-cx }$, $c>0$ we can continue with estimates for the Hilbert--Schmidt norm: 
\begin{align} \label{eq-HS1}
\begin{split}
\| \chi_\Omega R (z) \chi_\Sigma \|^2_{\mathrm{HS}} \leq &
C' \frac{d}{(2\pi )^2}  \cdot \frac{1 }{|z|^{1/2}} \int_{\Omega } \int_{\mathbb R}\frac{  \mathrm e^{-2\mathcal T_z (x_1^2+ \rho ^2)^{1/2}}}{   (x_1^2+ \rho ^2)^{1/2}}\,\mathrm{d}x_1\,\mathrm{d}y   \\ &= C'
\frac{d\,|\Omega |}{(2\pi )^2}  \cdot \frac{1 }{|z|^{1/2}} \int_{\mathbb R}\frac{  \mathrm e^{-2\mathcal T_z (x_1^2+ \rho ^2)^{1/2}}}{   (x_1^2+ \rho ^2)^{1/2}}\,\mathrm{d}x_1\,.
\end{split}
\end{align}
Applying again the monotonicity of the  integrand function 
and
the inequality
$$
\sqrt{2}(a^2+b^2)^{1/2 }\geq |a|+|b|\,,\quad \quad a\,,b\in \mathbb R\,,
$$
we arrive at
\begin{align}\label{eq-0}
\begin{split}
\| \chi_\Omega R (z) \chi_\Sigma \|^2_{\mathrm{HS}} \leq &  C'
\frac{\sqrt{2}d\,|\Omega |}{(2\pi )^2}  \cdot \frac{1 }{|z|^{1/2}}
\int_{\mathbb R}\frac{  \mathrm e^{-\sqrt 2\mathcal T_z (|x_1|+ \rho)}} {   |x_1|+ \rho }\,\mathrm{d}x_1
\\ &
\leq C\frac{|\Omega  |}{|z|^{1/2} \mathcal T_z}   \frac{\e^{-\sqrt 2 \mathcal T_z \rho }}{\rho } \,,
\end{split}
\end{align}
where $C= C'
\frac{d }{2\pi ^2} $.
This shows (\ref{eq-asymR0}).
\end{proof}

In fact the above reasoning shows that
the operator $\chi_\Sigma R (z) \chi_\Omega  \,:\, L^2 (\Omega  ) \to L^2 (\Sigma  )$ admits the same  Hilbert--Schmidt norm bound given by (\ref{le-normsbounds}).
Proceeding in the same way as in (\ref{eq-HS1})  one can show that $ \chi_{\Sigma }R(z)\chi_{\Sigma }$ is Hilbert--Schmidt.  Using the Fredholm theorem, cf.~\cite{RS1}, we can conclude that $(\mathrm I - \alpha \chi_{\Sigma }R(z)\chi_{\Sigma })^{-1}$ exists and it is bounded for $z\in [E_\alpha \,,\infty ) $ apart from a set of  isolated points.
Furthermore, applying (\ref{eq-defG}) and the boundedness of $(\mathrm I -\alpha \chi_{\Sigma }R(z)\chi_{\Sigma })^{-1}$ we obtain
the following  bound for operator $G_{\alpha, \beta } (z)\,:\,L^2 (\Omega )\to L^2 (\Omega )$,
    \begin{equation}\label{eq-asymG}
\| G_{\alpha, \beta } (z)\|_{\mathrm{HS}} \leq  C_1
\frac{ \e ^{-\sqrt{2}\mathcal T_z \rho } }{\rho}\,,
  \end{equation}
where $C_1$ is a positive constant which depends on $z$.


\begin{remark} Analytic continuation of the resolvent $R_{\alpha , \beta }(z)$. \rm{In Section~\ref{sec-analytic-extension} we constructed analytic continuation of $\chi_\Omega R_{\alpha, 0}(z)\chi_\Omega$ via the  segments $(E_{\alpha;k}, E_{\alpha;k+1})$. This means that $\mathrm I -\beta \chi_\Omega R_{\alpha, 0}(z)\chi_\Omega$ is analytic in $\mathcal C_k$. Moreover, applying the estimates (\ref{eq-asymR0}) and (\ref{eq-asymG}), we conclude that $\chi_\Omega R_{\alpha, 0} (z)\chi_\Omega$ is a Hilbert--Schmidt operator. Using again the Fredholm theorem, cf.~\cite{RS1}, we state that  $\mathrm I -\beta  \chi_\Omega R_{\alpha, 0}(z)\chi_\Omega$ is invertible in $\mathcal C_k$ apart from isolated  points. For any $g, f\in C_0 (\R^2)$ we can also construct the continuation of $gR_{\alpha , 0 }(z)f$ to $\mathcal C_k$ applying the reasoning from Section~\ref{sec-analytic-extension}. This leads to continuation
of
\begin{equation}\label{eq-resolvent-alpha-beta-IIsheet}
  f  R_{\alpha ,\beta  }(z)g= fR_{\alpha , 0} (z) g+\beta  fR_{\alpha, 0}(z) \chi_{\Omega  }(\mathrm I -\beta  \chi_{\Omega  }R_{\alpha ,0 }(z)\chi_{\Omega  })^{-1}\chi_{\Omega  } R_{\alpha, 0 }(z)g\,,
\end{equation}
via the segment $(E_{\alpha;k}, E_{\alpha;k+1})$ to $\mathcal C_k$ apart from isolated points.
}
\end{remark}
\bigskip

Let us back to the analysis of (\ref{eq-decompkernel}).
Combining the  decomposition (\ref{eq-birmanschwingerpoles}) and the fact that the operator valued function  $ z \mapsto A_n(z)$ is analytic a certain neighbourhood $\mathcal S_n$ of $\mathcal E_{\beta ;n }$  and taking into account  the   norm estimate
(\ref{eq-asymG}) we state that $\mathrm I -A_n (z)G_{\alpha, \beta }(z)$ is invertible for $\rho $ large enough. Consequently,  the operator contributing to (\ref{eq-decompkernel}) can be written
\begin{align} \label{eq-kercon1}
   \mathrm I - (\mathrm I- \beta \chi_\Omega R (z)\chi_\Omega)^{-1} G_{\alpha, \beta } (z)  = \Big[\mathrm I - A_n (z)G_{\alpha, \beta }(z)\Big]\,
  \Big[ \mathrm I + \frac{1}{\beta }\frac{1}{z-\mathcal E_{\beta ; n }}\Big[
  \mathrm I -A_n(z)G_{\alpha, \beta } (z)
  \Big]^{-1} \mathrm P_{\beta; n} G_{\alpha, \beta }(z)\Big]\,.
\end{align}
Applying   (\ref{eq-kercon1}) and (\ref{eq-decompkernel}) we come to the conclusion that the condition (\ref{eq-BSralpha-ext})
is equivalent to
\begin{align} \label{eq-kernelexp2}
    \ker \Big[ z-\mathcal E_{\beta ; n } + \frac{1}{\beta }\Big[\mathrm I -A_n(z)G_{\alpha, \beta } (z)\Big]^{-1}\mathrm  P_{\beta ;n} G_{\alpha, \beta }(z)\Big] \neq
    \emptyset \,.
\end{align}
Note that the operator $L_n (z ):= [\mathrm I -A_n (z)G_{\alpha, \beta } (z)]^{-1} \mathrm P_{\beta ;n} G_{\alpha, \beta }(z)$ acting in $L^2 (\Omega )$ is  rank one and it acts as
\begin{equation}\label{eq-Ln}
L_n (z)f= (w_{\beta ;n }\,, G_{\alpha, \beta }(z) f)_{L^2 (\Omega )}\big[\mathrm I -A_n (z)G_{\alpha, \beta } (z)\big]^{-1} w_{\beta ;n }\,.
\end{equation}
Therefore, the condition (\ref{eq-kernelexp2}) can be transited to the scalar equation
\begin{align} \label{eq-kernelexp3}
    z-\mathcal E_{\beta ; n } +  \frac{1}{\beta }
     (w_{\beta ;n }\,, G_{\alpha, \beta }(z) \big[\,\mathrm I -A_n(z)G_{\alpha, \beta } (z)\,\big]^{-1} w_{\beta ;n })_{L^2 (\Omega )}
=0\,.
\end{align}
With the above results, we are ready to prove the following statement, which establishes the resonances.
\begin{theorem} \label{th-resonances} Assume that $H_{0, \beta }$ has a simple, discrete eigenvalue  $\mathcal E_{\beta ; n} \in (E_{\alpha; j}\,, E_{\alpha; j+1})$.
Then the second sheet analytic continuation of $R_{\alpha, \beta }(z)$ defined by~(\ref{eq-resolvent-alpha-beta-IIsheet})
has a unique pole $z_n\in \mathcal S_n$,  i.e.  $\ker [\mathrm I - \beta \chi_\Omega R_{\alpha , 0}(z_n) \chi_\Omega ] \neq \emptyset $, with the asymptotics
\begin{equation} \label{ieq-asympI}
   z_{n}(\rho )= \mathcal E_{\beta ; n}+\mathcal O \Big(\frac{\e^{-\sqrt{2 |\mathcal E_{\beta ;n}|} \rho}  }{\rho }\Big)\,.
\end{equation}
\end{theorem}
\begin{proof} \noindent \emph{Step 1. Auxiliary statement:  the Neumann expansion.}  We adopt  the notation $b  := \frac{1}{\rho }$, which means that $b\to 0 $ for $\rho \to \infty $. Note that the operator $G_{\alpha , \beta }(z)$ depends on $b$ because it is defined by the embeddings  acting between spaces $L^2 (\Omega )$ and $L^2 (\Sigma )$, see~(\ref{eq-defG}). Since the parameter $b$ will play the essential rule in the further discussion we implement the notation
\begin{align} \label{eq-kernelexp5}
   \eta (z, b):=
     (w_{\beta ;n }\,, G_{\alpha, \beta }(z) [\,\mathrm I -A_n(z)G_{\alpha, \beta } (z)\,]^{-1} w_{\beta ;n })_{L^2 (\Omega )}\,.
\end{align} 
In view of (\ref{eq-asymG}) we have  $\lim_{b\to 0 }\eta (z,b) =0$ for any $z\in \mathcal S_n$. We continuously extend function $\eta $ assuming that $\eta (z, 0)=0$. Note that formula~(\ref{eq-kernelexp5}) can be rephrased in the terms of the  Neumann series
\begin{equation}\label{eq-Neuman}
[
  \mathrm I -A_n(z)G_{\alpha, \beta } (z)
  ]^{-1} =
  \mathrm I + A_n(z)G_{\alpha, \beta } (z) +[\,A_n(z)G_{\alpha, \beta } (z) \,]^2+...
\end{equation}
which implies 
\begin{align} \label{eq-kernelexp6}
\begin{split}
   \eta (z, b)=
     \Big[  & \,(w_{\beta ;n} \,, G_{\alpha, \beta }(z)  w_{\beta ;n})_{L^2 (\Omega )} +
(w_{\beta ;n}\, , G_{\alpha, \beta }(z)  A_n (z)G_{\alpha, \beta } (z) w_{\beta ;n})_{L^2 (\Omega )}
    \\ +&
      (w_{\beta ;n}\, , G_{\alpha, \beta }(z)  [A_n (z)G_{\alpha, \beta } (z) ]^2\, w_{\beta ;n}\,)_{L^2 (\Omega )}+...\Big]
=0\,.
\end{split}
\end{align}

\noindent \emph{Step 2. Unique solution of the spectral equation.} To discuss further properties of $\eta $ we will investigate  analycity of $G_{\alpha, \beta }(z)$, which is related to the resolvent $R_{\alpha, 0}(z) $ by the equivalence  $G_{\alpha, \beta }(z)= \beta \, [\,\chi_{\Omega }R_{\alpha, 0 }(z)\chi_{\Omega }- \chi_{\Omega }R(z)\chi_{\Omega }\,]$, cf.~(\ref{eq-defG}) and (\ref{eq-resolvent-alpha-beta}).
As a consequence of the first resolvent identity
\begin{equation}\label{eq-residen}
\chi_{\Omega } [\,R(z)-R(z')\,]\chi_{\Omega } = (z'-z)\chi_{\Omega } R(z)R(z')\chi_{\Omega }\,,
\end{equation}
and the analogous identity  for $R_{\alpha, 0} (z)$, we get
\begin{align} \label{eq-expansionG}
\begin{split}
G_{\alpha, \beta }(z) &- G_{\alpha, \beta }(z') = 
\beta
 \,\chi_{\Omega }\big [\,R_{\alpha, 0 }(z)-
R_{\alpha, 0 }(z')-
 R(z) +R(z')\,\big]\chi_{\Omega }\\
& = (z'-z)
\beta
 \,\big [ \chi_{\Omega }\,R_{\alpha, 0 }(z)R_{\alpha, 0 }(z')\chi_{\Omega } - \chi_{\Omega }\,R(z)R (z')\chi_{\Omega } \,\big]
\\ &=(z'-z) \big[\chi_{\Omega }R(z)\chi_{\Omega }G_{\alpha, \beta } (z') + G_{\alpha, \beta } (z) \chi_{\Omega }R(z') \chi_{\Omega }+\frac{1}{\beta } G_{\alpha, \beta } (z) G_{\alpha, \beta } (z')  \big] \,.
\end{split}
\end{align}
Note that the operators in the last line of (\ref{eq-expansionG}) can be bounded by means of~(\ref{eq-asymG}). In view of the  analyticity of $A_n (z)$, the same  reasoning can be proceeded for operators appearing  in further terms of (\ref{eq-kernelexp6}).
Therefore we can state that $\eta (z,b) $ is analytic as the function of $z\in \mathcal S_n$, and
\begin{equation}
    \eta (z, b)- \eta (z', b)=(z-z') f(z,z';b)\,,
\end{equation}
    where $f(\cdot ,\cdot\, ; b  )$ is analytic and   $\lim _{b\to 0 }f(z, z'; b)=0$ uniformly with  respect to $z$ and $z'$. This gives $\partial_z \eta (z;b)|_{b= 0}=0$, where we understand the symbol $|_{b= 0}$ as the limiting value for $b\to 0$. The latter follows directly from the estimate (\ref{eq-asymG}) and eq.~(\ref{eq-expansionG}). This means that
\begin{equation} \label{eq-deriv=1}
\frac{\partial \tilde \eta (z,b)}{\partial z }\Big |_{b=0}    =1\,,\quad
\mathrm{where}\quad
\tilde \eta (z,b)
 :=z-\mathcal E _{\beta; n } +\beta ^{-1}\eta (z,b )
\,.
\end{equation}
Note that the spectral equation (\ref{eq-kernelexp3}) reads
\begin{equation} \label{eq-spectralreson}
 \tilde \eta (z,b ) = 0\,.
\end{equation}
Combining (\ref{eq-deriv=1})  with the implicit  function theorem and the fact that $ \tilde \eta (\mathcal E_{\beta ; n  }, 0 )=0 $, we can conclude  that for $b$ small enough the equation  (\ref{eq-spectralreson}) has a unique solution $z_n (b)$ such that
$\tilde \eta (z_n (b), b ) = 0 $. The asymptotics of the solution can be recovered from (\ref{eq-kernelexp6}). Since we are interested in a solutions for $b\neq 0$ we again back  to the variable $\rho $ and the asymptotics  $\rho\to \infty$. Without a danger of confusion, we keep the same notation for $z_n (\cdot ) $ expressed in the terms of $\rho$.

\noindent \emph{Step 3. Asymptotics of the solution.} Combining the spectral equation (\ref{eq-spectralreson}) and the expansion (\ref{eq-kernelexp6}) together with the definition of $\tilde \eta (\cdot, \cdot)$ from~(\ref{eq-deriv=1}), we get
\begin{align} \begin{split}\label{eq-expansion1}
z_n (\rho)= \mathcal E_{\beta ;n } - \frac{1}{\beta } \Big[  & \,(w_{\beta ;n} \,, G_{\alpha, \beta }(z)  w_{\beta ;n})_{L^2 (\Omega )} +
(w_{\beta ;n}\, , G_{\alpha, \beta }(z)  A_n (z)G_{\alpha, \beta } (z) w_{\beta ;n})_{L^2 (\Omega )}
    \\ + &
      (w_{\beta ;n}\, , G_{\alpha, \beta }(z)  [A_n (z)G_{\alpha, \beta } (z) ]^2\, w_{\beta ;n}\,)_{L^2 (\Omega )}+...\Big]\,.
      \end{split}
\end{align}
As follows from the above discussion, the dominated term of  $z_n (\rho)$  comes from the term $-\beta ^{-1}(w_{\beta ;n} \,, G_{\alpha, \beta }(z)  w_{\beta ;n})_{L^2 (\Omega )}$, $z\in \mathcal S_n$.
Therefore we expand
\begin{equation}\label{eq-expansion2}
(w_{\beta ;n} \,, G_{\alpha, \beta }(z)  w_{\beta ;n})_{L^2 (\Omega )}= (w_{\beta ;n} \,, G_{\alpha, \beta }(\mathcal E_{\beta ; n })  w_{\beta ;n})_{L^2 (\Omega )}+ (z-\mathcal E_{\beta; n}) (w_{\beta ;n} \, T_n (z)w_{\beta ;n})_{L^2 (\Omega )}\,,
\end{equation}
where $T_n (z)$ has the norm that can be bounded by the norm of $G_{\alpha, \beta } (z)$, see (\ref{eq-expansionG}).
Applying the expansions (\ref{eq-expansion1}) and (\ref{eq-expansion2}) we get
\begin{align}
\begin{split}
\label{eq-expansion3}
z_n (\rho)= &\mathcal E_{\beta ;n } - \frac{1}{\beta }\Big [\, 1+\beta^{-1} (w_{\beta ;n} \,, T_n (z)w_{\beta ;n})_{L^2 (\Omega )}\,\Big ]^{-1} \Big[   \,(w_{\beta ;n} \,, G_{\alpha, \beta }(\mathcal E_{\beta; n })  w_{\beta ;n})_{L^2 (\Omega )} \\
&+
(w_{\beta ;n}\, , G_{\alpha, \beta }(z)  A_n (z)G_{\alpha, \beta } (z) w_{\beta ;n})_{L^2 (\Omega )}
      +(w_{\beta ;n}\, , G_{\alpha, \beta }(z)  [A_n (z)G_{\alpha, \beta } (z) ]^2\, w_{\beta ;n}\,)_{L^2 (\Omega )}+...\Big]\,.
   \end{split}
\end{align}
The above expansion shows that the dominated term is determined by
\begin{equation}\label{eq-domin}
-\frac{1}{\beta }(w_{\beta ;n} \,, G_{\alpha, \beta }(\mathcal E_{\beta; n })  w_{\beta ;n})_{L^2 (\Omega )}\,,
\end{equation}
which, in view of (\ref{eq-asymG}),
yields the asymptotics
\begin{equation}\label{eq-zn}
z_{n} (\rho )= \mathcal E_{\beta ; n}+\mathcal O \Big(\frac{\e^{-\sqrt{2 |\mathcal E_{\beta; n}|} \rho}  }{\rho }\Big)\,.
\end{equation}
This completes the proof.
\end{proof}

\section{Fermi's golden rule} \label{sec-Fermi}

The discussion of this section is devoted to a deeper analysis of the asymptotitcs of  $\Im z_n (\rho)$.
 First we analyze the component  of $\Im z_n (\rho)$ given by
\begin{equation} \label{eq-domin1}
\Gamma_n (\rho ):=-\frac{1}{\beta }\Im
(w _{\beta ; n} , G_{\alpha, \beta }(\mathcal E_{\beta; n })  w_{\beta ; n})_{L^2 (\Omega )} \,,
\end{equation}
which comes from (\ref{eq-domin}); cf.~(\ref{eq-expansion3}).
Using the Birman-Schwinger principle (\ref{eq-birmanschwingeromega})   we get
\begin{equation}\label{eq-resonances-aux2}
\chi_{\Omega }R(\mathcal E_{\beta; n } )\chi_\Omega  w_{\beta; n }=\chi_\Omega \omega _{\beta ; n }\,,\quad \chi_{\Omega }R(\mathcal E_{\beta; n }) \chi_\Omega  w_{\beta; n}=  \beta^{-1}w_{\beta; n}\,.
\end{equation}
This yields
\begin{equation}\label{eq-resonances-aux3}
    (w _{\beta ; n} , G_{\alpha, \beta }(\mathcal E_{\beta ;n})  w_{\beta ; n} )_{L^2 (\Omega )} = \beta^2 (\chi_\Omega \omega _{\beta ; n } , G_{\alpha, \beta } (\mathcal E_{\beta ;n}) \chi_\Omega \omega _{\beta ; n })_{L^2 (\Omega )}\,.
\end{equation}
On the other hand,  $\mathcal E_{\beta; n}<0$ belongs to the resolvent sets of  $ R_{\alpha, 0}^c(\cdot ) $ and $R(\cdot )$, which are self adjoint operators ($R_{\alpha, 0}^c(\cdot ) $ is defined by (\ref{eq-Rdc})). Therefore,  we have
\begin{equation}\label{eq-resonances-aux7} \Im (\chi_\Omega \omega_{\beta ; n } , R(\mathcal E_ {\beta;n }) \chi_\Omega \omega_{\beta ; n })_{L^2 (\Omega )}=\Im (\chi_\Omega \omega_{\beta ; n } , R_{\alpha, 0}^c(\mathcal E_ {\beta;n }) \chi_\Omega \omega_{\beta ; n })_{L^2 (\Omega )}=0\,.\end{equation}
Using  again the equation 
\begin{equation}\label{eq-Gbeta}
   G_{\alpha, \beta }(z)= \beta [\,\chi_{\Omega }R_{\alpha, 0 }(z)\chi_{\Omega }- \chi_{\Omega }R(z)\chi_{\Omega }\,]\,,
\end{equation}
together with~(\ref{eq-resonances-aux7}),
 we get
\begin{equation}\label{eq-resonances-aux4}
\Im (\chi_\Omega \omega_{\beta ; n } , G_{\alpha, \beta } (\mathcal E_{\beta;n} ) \chi_\Omega \omega_{\beta ; n })_{L^2 (\Omega )}= \beta \Im  (\chi_\Omega \omega_{\beta ; n } , R_{\alpha, 0}^d(\mathcal E_ {\beta;n }) \chi_\Omega \omega_{\beta ; n })_{L^2 (\Omega )}\,.
\end{equation}
Furthermore, applying~(\ref{eq-Rd}) we have
\begin{align}
\begin{split} \label{eq-epsilonlimit1}
\Im  (\chi_\Omega &\omega_{\beta ; n } , R_{\alpha, 0}^d(\mathcal E_ {\beta;n }) \chi_\Omega \omega_{\beta ; n })_{L^2 (\Omega )} \\ &=
\Im \frac{1}{2\pi}\Big[\lim_{\varepsilon\to 0 }\sum_{k\in \mathcal N_\alpha} \int_{\Omega \times \Omega  } \int_{\mathbb R} \frac{\e^{ip_1(x_1-y_1)}}{p_1^2+E_{\alpha; k} -(\mathcal E_ {\beta;n } +i\varepsilon )}  \phi_{\alpha;k }(x_2) \phi_{\alpha;k}(y_2) ^\ast \omega_{\beta ; n } (x)^\ast \omega_{\beta ; n } (y) \mathrm d  p_1\,\mathrm d x \,\mathrm d y \Big] \,.
\end{split}
\end{align}
 To deal with the integral
\begin{equation}\label{eq-frac}
\int_{\mathbb R} \frac{\e^{ip_1(x_1-y_1)}}{p_1^2-(b  _{k,n}+i\varepsilon )} \mathrm d  p_1\,,\quad \mathrm{where }\quad  b_{k,n}:= -(E_{\alpha; k} -\mathcal E_ {\beta;n })\,,
\end{equation}
we distinguish two cases  $b_{k,n}>0$ and $b_{k,n}<0$. Assume first that $b_{k,n}>0$; then the above integrand has two poles $\pm (\tilde b_\epsilon+i \tilde \varepsilon )$, where $\tilde b_\epsilon \to{ \sqrt {b_{k,n}}} $ and $\tilde \varepsilon \to  $ for $\varepsilon \to 0$.
Applying  the
Sokhotski–Plemelj formula which can be
formally written $\lim_{\varepsilon\to 0 }\frac{1}{x\pm i\varepsilon }=\mathcal P \frac{1}{x} \mp i\pi \delta (x)$, where $\mathcal P$ stands for the principle value,
we get
\begin{align}
\begin{split}\label{eq-epsilonlimit0}    \lim_{\varepsilon\to 0 }\int_{\mathbb R} \frac{\e^{ip_1(x_1-y_1)}}{p_1^2-(b_{k,n}  +i\varepsilon )} \mathrm d  p_1 =&
\lim_{\varepsilon\to 0 }\int_{\mathbb R} \frac{\e^{ip_1(x_1-y_1)}}{(p_1-(\tilde  b_\varepsilon   +i\tilde \varepsilon ))(p_1+(\tilde  b_\varepsilon   +i\tilde \varepsilon ))} \mathrm d  p_1 \\ =&
\frac{i\pi }{2\sqrt{b_{k,n}} } \Big[ \mathrm e ^{i\sqrt{b_{k,n}}(x_1 - y_1) } +  \mathrm e ^{-i\sqrt{b_{k,n}} (x_1 - y_1) }   \Big] + \mathcal P \int_{\mathbb R} \frac{\e^{ip_1(x_1-y_1)}}{p_1^2-b_{k,n}  } \mathrm d  p_1
\\ =&\frac{i\pi}
{\sqrt{b_{k,n}}} \cos (\sqrt{b_{k,n}}|x_1-y_1|)+ \mathcal P \int_{\mathbb R} \frac{\e^{ip_1(x_1-y_1)}}{p_1^2-b_{k,n}   } \mathrm d  p_1 \,.
\end{split}
\end{align}
For $b_{k,n}<0$ the  analogous limit reproduces $\pi \frac{\mathrm{exp}(-|b_{k,n}|^{1/2}| x_1-x_2| )}{|b_{k,n}|^{1/2}}$, which contributes to the real component of (\ref{eq-resonances-aux4}).
This, in view of (\ref{eq-epsilonlimit1}) and (\ref{eq-epsilonlimit0}) implies
\begin{align}
\begin{split}\label{eq-imresolvent}
\Im  (\chi_\Omega \omega_{\beta ; n }& , R_{\alpha, 0}^d(\mathcal E_ {\beta;n }) \chi_\Omega \omega_{\beta ; n })_{L^2 (\Omega )} \\ &=
\sum_{k \,:\, E_{\alpha ; k}< \mathcal E_{\beta ;n } }
\frac{1}{  2\sqrt{b_{k,n}}}
\int_{\Omega \times \Omega  } \cos (\sqrt{b_{k,n}}|x_1 - y_1|)
 \phi_{\alpha;k }(x_2) \phi_{\alpha;k}(y_2) ^\ast \omega_{\beta ; n } (x)^\ast \omega _{\beta ; n } (y) \,\mathrm d x \,\mathrm d y \\  &=
\sum_{k=1} ^j
\frac{1}{  2\sqrt{ \mathcal E_ {\beta;n }- E_{\alpha; k}  } }
\int_{\Omega \times \Omega  } \cos ( \sqrt{ \mathcal E_ {\beta;n }- E_{\alpha; k}  }|x_1 - y_1|)
 \phi_{\alpha;k }(x_2) \phi_{\alpha;k}(y_2)^\ast \omega_{\beta ; n } (x)^\ast \omega _{\beta ; n } (y) \,\mathrm d x \,\mathrm d y\,, \end{split}\end{align}
where, in the last step, we have  changed the summation from $\sum_{k \,:\, E_{\alpha ; k}< \mathcal E_{\beta ;n } }$ to $\sum_{k=1} ^j$
since $\mathcal E_{\beta ; n} \in (E_{\alpha; j}\,, E_{\alpha; j+1})$.
Using now (\ref{eq-domin1}),   
(\ref{eq-resonances-aux3}), (\ref{eq-resonances-aux4})  and (\ref{eq-imresolvent}) we can conclude
\begin{align}
\begin{split} \label{leq-fermi-special0}
\Gamma _n(\rho ) =-\beta^2\sum_{k=1} ^j
\frac{1}{  2\sqrt{ \mathcal E_ {\beta;n }- E_{\alpha; k}  } }
\Big[\int_{\Omega \times \Omega  }& \cos ( \sqrt{ \mathcal E_ {\beta;n }- E_{\alpha; k}  }|x_1 - y_1|)
 \phi_{\alpha;k }(x_2) \phi_{\alpha;k}(y_2) ^\ast  \\ & \times \omega_{\beta ; n } (x)^\ast \omega _{\beta ; n } (y) \,\mathrm d x \,\mathrm d y \Big]\,.
 \end{split}
\end{align}
In fact, $\Gamma_n (\rho)$ can be expressed by means of the generalized eigenfunctions $\varphi_{k} (x, p_1)= \frac{1}{\sqrt{2\pi}}\e^{i p_1 x_1 }\phi_{\alpha ; k}(x_2)$ of $H_{\alpha, 0}$.
Using  the first term of the second line of  (\ref{eq-epsilonlimit0}) and $\varphi_{k}(x, p_1)$  we can write  
\begin{align}
\begin{split} \label{eq-Gamman1}
\Gamma _n (\rho )=  -\beta ^2 \sum_{k=1}^j\frac{\pi}{ 2\sqrt{ \mathcal E_ {\beta;n }- E_{\alpha; k}  }
}\Big[ &\Big| \int_{\Omega }\omega_ {\beta; n } (x)^{\ast } \varphi_{k} (x,p_1 )\mathrm d x \Big|^2_{p_1=(\mathcal E_{\beta; n}- E_{\alpha; k } )^{1/2 }} \\ +
&\Big| \int_{\Omega }\omega_ {\beta; n } (x)^{\ast } \varphi_{k} (x,p_1 )\mathrm d x \Big|^2_{p_1= -(\mathcal E_{\beta; n}- E_{\alpha; k } )^{1/2 }} \,\Big]\,,
\end{split}
    \end{align}
    where the lower indexes indicate that the values of functions are taken at specified points.
The above Formula 
  reveals the well known Fermi's golden rule, cf.~\cite[XII.~6]{RS4}. In Appendix~\ref{sec-appendix}, we show that the further terms contributing to (\ref{eq-expansion3}) possess the stronger asymptotics then (\ref{eq-Gamman1}).
\begin{remark}\rm{
 Note that, $\varphi_{k} (\cdot ,p_1 )\notin L^2 (\mathbb R^2)$, however, the function $\chi_{\Omega  }\varphi_{k} (\cdot ,p_1 )$ is well defined as an element of $ L^2 (\mathbb R^2)$. The expressions (\ref{eq-Gamman1}) can be reformulated in the form more familiar to physicists
  \begin{align}
  \begin{split} \label{eq-fermi-special-special}
\Gamma _n (\rho ) = -\beta ^2\sum_{j=1}^j \frac{\pi}{2 \sqrt{ \mathcal E_ {\beta;n }- E_{\alpha; k}  }
} &\Big [ |  ( \omega_ {\beta; n } , \chi_{\Omega  }\varphi_{k} (\cdot ,p_1 ) \,)  |^2 _{p_1=(\mathcal E_{\beta; n}- E_{\alpha; k } )^{1/2 }} \\ &+ |  ( \omega_ {\beta; n } , \chi_{\Omega  }\varphi_{k} (\cdot ,p_1 ) \,)  |^2 _{p_1=-(\mathcal E_{\beta; n}- E_{\alpha; k } )^{1/2 }} \Big]\,,
    \end{split}
    \end{align}
where 
the scalar product in the above formulae is defined in the space $L^2 (\mathbb R^2)$.
The resonance width is determined  by
  \begin{equation}\label{eq-gamma}
  \gamma_n (\rho ) := 2  \, |\Im z_n (\rho )|\,,
  \end{equation}
 and, therefore  the dominated term of $ \gamma_n (\cdot )$  takes the form
  \begin{equation}\label{eq-gamma2}
  \gamma_n(\rho ) \approx  2 \pi
  \sum_{k=1}^j \frac{1}{ 2\sqrt{ \mathcal E_ {\beta;n }- E_{\alpha; k}  }
}
\Big[ | ( \omega_ {\beta; n } , V_{\beta }\varphi_{k} (\cdot ,p_1 ) \,)  |^2 _{p_1= p_{n,k}}
+
| ( \omega_ {\beta; n } , V_{\beta }\varphi_{k} (\cdot ,p_1 ) \,)  |^2 _{p_1= -p_{n,k}} \Big]\,,
\end{equation}
where $V_{\beta } = - \beta \chi_\Omega $ and $ p_{n,k}= (\mathcal E_{\beta; n}- E_{\alpha; k })^{1/2} $.
The above formula represents also
the decay rate of the resonant state  in the first approximation.  It reflexes the characteristic properties, namely the contribution of the transition amplitude $ ( \omega_ {\beta; n } , V_{\beta }\varphi_{k} (\cdot ,p_1 ) \,)$ and spectral density function
$(2\sqrt{ \mathcal E_ {\beta;n }- E_{\alpha; k}  })^{-1}$. The summation $ \sum_{k=1}^j $ comes from the fact that we have to take into account all modes which `catch' the resonant energy $\mathcal E_{\beta;n }$.
 }\end{remark}

\subsection{Asymptotics of $\Gamma_n (\rho )$}
Let us analyze in more detail the asymptotic behaviour of $\Gamma _n (\rho )$ for $\rho \to \infty$. For this aim we use the expression (\ref{eq-imresolvent}). Using the fact that  for  $x_2 \in \mathbb R\setminus [0, d]$, we have
\begin{equation} \label{eq-formphi}
\phi_{\alpha ; k } (x_2)=M_{\alpha, k} \e^{- \sqrt{|E_{\alpha; k}|}{|x_2| }}\,,
    \end{equation}
where $M_\alpha $ is the normalization constant, we can estimate (\ref{eq-imresolvent}) by
\begin{align}
\begin{split} \label{eq-sum}
\sum_{k=1}^j
\frac{|M_{\alpha , k}|^2 \e ^{- 2\sqrt{|E_{\alpha; k}|}{\rho}} }{  2\sqrt{ \mathcal E_ {\beta;n }- E_{\alpha; k}  } }&
\int_{\Omega \times \Omega  } \big|\cos ( \sqrt{ \mathcal E_ {\beta;n }- E_{\alpha; k}  }|x_1 - y_1|) \,
 \omega_{\beta ; n } (x)^\ast \omega _{\beta ; n } (y) \big | \,\mathrm d x \,\mathrm d y \\  &\leq j C_{j,n}  \e ^{- 2\sqrt{|E_{\alpha; j} |} \rho }\,,
 \end{split}
\end{align}
where
\begin{equation}\label{eq-C}
C_{j,n} :=  \mathrm{max }_{1\leq k\leq j } \Big[\frac{    |M_{\alpha, k }|^2
}{  2\sqrt{ \mathcal E_ {\beta;n }- E_{\alpha; k}  } }
\int_{\Omega \times \Omega  }\big| \cos ( \sqrt{ \mathcal E_ {\beta;n }- E_{\alpha; k}  }|x_1 - y_1|) \,
 \omega_{\beta ; n } (x)^\ast \omega _{\beta ; n } (y) \big|\,\mathrm d x \,\mathrm d y \Big]\,;
\end{equation}
in the above reasoning we used  the fact that $\sqrt {|\mathcal E_{\beta ; n }|} <
\sqrt { |E_{\alpha  ; k }|} $ for any $k=1,...,j$. In fact it follows from  (\ref{eq-Gamman1}) that all the  components labeled by $1,..,j$ contributing to  $\Gamma_n (\rho)$  have the  negative sign (they can  not  vanish because the functions in the scalar products are not orthogonal).  Using (\ref{eq-formphi}) we state
\begin{equation} \label{eq-asympII}
    \Gamma_n (\rho )\approx 
    \, \e ^{- 2\sqrt{|E_{\alpha; j} |} \rho }\,,\quad \mathrm{where }\quad\mathcal E_{\beta ; n} \in (E_{\alpha; j}\,, E_{\alpha; j+1})\,.
\end{equation}
On the other hand, note that for $\rho$ sufficiently large we have
\begin{equation}\label{eq-sum11}
\e ^{- 2\sqrt{|E_{\alpha; j} |} \rho } \leq \frac{\e ^{- \sqrt{2|\mathcal E_{\beta ; n} |} \rho } }{\rho }\,;
\end{equation}
therefore  the asymptotics (\ref{eq-asympII}) enhances
(\ref{ieq-asympI}) for the imaginary component of $z_n (\rho)$.

\section{ Generalization of the system: varying transverse profile and perturbation defined by the Kato class measure} \label{sec-Kato}

In this section we consider a generalization of the model and assume that the transverse profile of the waveguide is given by a piecewise continuous function $\alpha\,:\, [0, d ]\to \mathbb R$. On the other hand, we generalize the  well potential supported on $\Omega$ by a  Radon  measure $\mu $ belonging to the Kato class, i.e. satisfying
\begin{equation}
    \lim_{\varepsilon \to 0 } \,\sup _{x\in \mathbb R^2}\int_{B(x,\varepsilon ) } |\log (|x-y|)|\,\mathrm d \mu (y)=0\,,
\end{equation}
where  $B(x,\varepsilon ) \subset \mathbb R^2$ denotes a ball centered at $x$ and with the radius $\varepsilon$.
Throughout this section we assume that  the support $\Omega \subset \mathbb R^2 $ of $\mu$ is compact; in this part of further discussion where support of $\mu$  plays an essential  role we will use the notation  $\mu_\Omega \equiv \mu$. The examples of the Kato measures has been discussed in~Introduction.



To maintain the consistency with previous discussion we define the remote perturbation by $-\beta \mu $, where $\beta \in \mathbb R$.
Then the Hamiltonian is  associated to  the sesquilinear form
\begin{equation}\label{eq-formalphabeta}
     \mathcalligra { E}_{\alpha, \beta \mu} [f]= \int_{\mathbb R^2 } |\nabla f (x) |^2 \,\mathrm d x- \int_{\Sigma }\alpha(x_2)|f (x)|^2 \,\mathrm d x
-\beta \int_{\mathbb R^2 }|I_\mu f (x)|^2 \,\mathrm d \mu (x)\,, \quad f\in W^{1,2 } (\R^2 )\,,
\end{equation}
where $I_\mu$ denotes the space embedding $I_\mu \,:\,  W^{1,2 } (\R^2 ) \to   L^2 _\mu \equiv L^2 (\mathbb R^2, \,\mathrm d \mu)$ which is continuous, cf.~\cite{AF2003}.
For any Kato class measure $\mu$ and  any small parameter $\varepsilon >0$ there exists $C(\varepsilon)$ such that
\begin{equation}\label{eq-estImu1}
\int_{\mathbb R^2 } |I_\mu f(x)|^2 \mathrm d \mu (x) \leq \varepsilon \int_{\mathbb R^2 } |\nabla f(x)|^2 \mathrm d x
+ C(\varepsilon ) \int_{\mathbb R^2 } | f(x)|^2 \mathrm d x\,,
\end{equation}
cf.~\cite{BEKS}.   This yields
\begin{equation}\label{eq-estImu2}
\int_{\mathbb R^2 } |I_\mu f(x)|^2 \mathrm d \mu (x) \leq \varepsilon \Big[\int_{\mathbb R^2 } |\nabla f(x)|^2 \mathrm d x
-  \int_{\Sigma  } \alpha (x_2)| f(x)|^2 \mathrm d x \Big] + C_1(\varepsilon ) \int_{\mathbb R^2 } | f(x)|^2 \mathrm d x\,,
\end{equation}
where $C_1 (\varepsilon)= C(\varepsilon) +\varepsilon \sup\, |\alpha (x_2)|$, which implies that the  operator associated to $\mathcalligra E_{\alpha, \beta \mu }$  has a unique self adjoint extension in view of the KLMN Theorem, cf.~\cite{RS1}.
The resulting  operator defines   a generalization of $H_{\alpha, \beta  }$ given by (\ref{eq-Hamiltonianalphabeta}).

To reconstruct the resolvent of $H_{0, \beta \mu }$ we define the embeddings of $R(z)$ to $ L^2 _\mu $ by
\begin{equation}\label{eq-Restim}
R(z)|_{0 \mu}f= R(z)\ast f\mu = \frac{1}{2\pi}\int_{\mathbb R^2 } K_0 (k_z|x-y|)f(y)\, \mathrm d \mu (y )\,:\, L^2_\mu \to L^2 (\mathbb R^2)\,,
\end{equation}
where $k_z= -i\sqrt z $.
Furthermore, we define  $R(z)|_{ \mu 0}\,:\,L^2 (\mathbb R^2)\to L^2 _\mu $ which acts as
$ R(z)|_{ \mu 0 }f= I_\mu R(z) f\,$
and, finally, we construct  the bilateral embedding by
\begin{equation}\label{eq-Restim0}
R(z)|_{\mu \mu} := I_\mu R(z)|_{0 \mu}\,:\, L^2_\mu \to L^2_\mu\,.\end{equation}
 We are ready to construct the resolvent of $(H_{0, \beta }- z )^{-1}$ by
\begin{equation} \label{eq-betaresolvent1gen}
   R_{0,\beta \mu }(z)= R (z) +\beta  R(z)|_{0\mu} (\mathrm I -\beta  R(z)|_{\mu \mu})^{-1} R (z)|_{\mu 0} \,.
\end{equation}
The similar types of Hamiltonian were discussed, for example in~\cite{BEKS, KL}.

On the other hand, the resolvent  $R_{\alpha, 0} (z)$ represented  in the Birman--Schwinger form  is given by
\begin{equation} \label{eq-BSralpha-gen}
    R_{\alpha, 0 }(z)= R(z)  + R(z) \alpha^{1/2}\chi_{\Sigma }(\mathrm I -|\alpha|^{1/2}  \chi_{\Sigma }R(z) \alpha^{1/2}\chi_{\Sigma })^{-1}|\alpha|^{1/2}\chi_{\Sigma } R(z)\,,
\end{equation}
where $\alpha (\cdot)$ is now a function of $x_2$ and $\alpha^{1/2} = \mathrm{sgn}(\alpha ) |\alpha|^{1/2} $. 
Using the embeddings of $R(z)$ to $L^2_\mu $ and the strip resolvent~(\ref{eq-BSralpha-gen}) we can build the operator
$$
\mathrm I- \beta R_{\alpha , 0 }(z)|_{\mu \mu}= \mathrm I- \beta \Big[ R(z)|_{\mu \mu }+  R(z) |_{\mu 0}\alpha^{1/2}\chi_\Sigma (\mathrm I -|\alpha|^{1/2} \chi_\Sigma R(z) \alpha^{1/2} \chi_\Sigma )^{-1}|\alpha|^{1/2}\chi_\Sigma   R(z)|_{0\mu }\Big]\,:\,L^2_\mu \to L^2_\mu\,,
$$
which contributes to
the Birman--Schwinger resolvent of $H_{\alpha , \beta }$.
Assuming now that $z$ is such that  the above operator is invertible,
 we construct the  resolvent  of $H_{\alpha, \beta }$ by
\begin{equation}\label{eq-resolvent-alpha-beta1}
    R_{\alpha ,\beta \mu }(z)= R_{\alpha , 0} (z) +\beta  R_{\alpha, 0}(z) |_{ 0 \mu}(\mathrm I -\beta  R_{\alpha ,0 }(z)|_{\mu \mu})^{-1} R_{\alpha, 0 }(z)|_{\mu 0}\, :\,L^2 (\mathbb R^2) \to L^2 (\mathbb R^2)\,.
\end{equation}
The stability of the essential spectrum
\begin{equation}\label{eq-stab}
\sigma _{\mathrm {ess} }(H_{\alpha ,\beta \mu }) =
\sigma _{\mathrm{ess} }(H_{\alpha ,0  }) = [E_\alpha, \infty\,)
\end{equation}
is a consequence of  the  following reasoning. Note that the operators $R(z)|_{\mu 0}\,:\,L^2 (\mathbb R^2)\to L^2 _\mu $ and $R(z)|_{0\mu }\,:\,L^2 _\mu \to  L^2(\mathbb  R^2)$ are compact in view of~\cite[Le.~3.1]{BEKS}. Therefore, using~(\ref{eq-BSralpha-gen})
we can state that both
\begin{equation}\label{eq-emmbe}
R_{\alpha, 0}(z) |_{ 0 \mu}\,:\, L^2_\mu \to L^2 (\mathbb R^2)\quad \mathrm{and} \quad  R_{\alpha, 0}(z) |_{  \mu0}\,:\, L^2 (\mathbb R^2)\to L^2_\mu
\end{equation}
are also compact since the operator $(\mathrm I -|\alpha|^{1/2}  \chi_{\Sigma }R(z) \alpha^{1/2}\chi_{\Sigma })^{-1}$ is bounded. Consequently, we get the compactness of the second component of~(\ref{eq-resolvent-alpha-beta1}) and, in view of the Weyl's theorem, we come to the conclusion that the essential spectra of $H_{\alpha, 0}$ and $H_{\alpha, \beta \mu }$
coincide.
\\ \\
\noindent \emph{Assumptions on $H_{\alpha, 0}$ and $H_{0, \beta \mu }$.}  As before, let $\mathrm h_\alpha $ stand for the transverse component of the Hamiltonian $H_{\alpha, 0 }$. In the current framework it takes the form
\begin{equation} \label{eq-trasvHam-gen}
    \mathrm h_\alpha = -\frac{d^2}{dx_2^2} - \alpha (x_2)\chi_{[0, d]} (x_2) \,:\, W^{2,2} (\mathbb R)\to L^2 (\mathbb R)\,.
\end{equation}
In the following we assume that $\mathrm h_\alpha$ has at least one negative discrete spectrum point. It happens, for example, if $\int_{[0,d]}\alpha (x_2)\mathrm d x _2 \geq 0$, cf.~\cite{Simon-weakly}. Following the convention of the previous discussion, we will keep the  notation $\{E_{\alpha ; 1}\leq E_{\alpha ; 2}\leq ...\leq E_{\alpha ; N_\alpha }\}$ and $\{\phi_{\alpha ; n}\}_{n\in {\mathcal N_\alpha }}$ for the eigenvalues and eigenfunctions of $\mathrm h_\alpha $, see.~(\ref{eq-evaluesalpha}).
 On the other hand, the eigenvalues and eigenfunctions of $H_{0, \beta \mu}$ will be denoted by
 $\{\mathcal E_{\beta ; 1}\leq  \mathcal E_{\beta ; 2} \leq ...\leq \mathcal E_{\beta ; N_\beta }\}$ and
$ \{\omega_{\beta  ; n}\}_{n\in {\mathcal N_\beta  }}$, cf.~(\ref{eq-eigenvalubeta}). In the following we assume that there exists $n\in \mathcal N_\beta $ and $j\in \mathcal N_\alpha $ such that
$\mathcal E_{\beta ; n} \in (E_{\alpha; j}\,, E_{\alpha; j+1})$.
\\ \\
To investigate  the problem of resonances in the generalized potential  system
we  need to estimate
\begin{equation}\label{eq-Galphabeta0}
G_{\alpha, \beta }(z)= \beta  R_{\alpha, 0}(z) |_{ \mu \mu}(\mathrm I -\beta  R_{\alpha ,0 }(z)|_{\mu \mu})^{-1} R_{\alpha, 0 }(z)|_{\mu \mu}\, :\,L^2 _\mu \to L^2 _\mu\,,
\end{equation}
which states an  analogue of~(\ref{eq-defG}). Similarly as before, we start from the estimates of the embeddings of the free resolvent
\begin{align}\label{eq-asymR0-gen}
\begin{split}
\|  R (z)|_{\mu 0  } \chi_\Sigma \|^2_{\mathrm{HS}} &=
\frac{1}{2\pi }\int_{\mathbb R^2  }\int_{\Sigma }|I_\mu K_0 (k_z |x-y|) |^2 \mathrm d \mu (x)  \mathrm d y \\ &\leq
 \sup_{x\in \Omega }\frac{|\mu |}{2\pi }\int_{\Sigma }|I_\mu K_0 (k_z |x-y|) |^2 \mathrm d y \leq C\frac{ |\mu |\ }{|z|^{1/2} \mathcal T_z}\, \frac{\e ^{-\sqrt{2}\mathcal T_z\rho } }{\rho}\,,
 \end{split}
\end{align}
where $|\mu|= \int_{\mathbb R^2 } \mathrm d \mu$ and the embedding $I_\mu $ is applied to the kernel defined by $K_0$ understood as the function of the variable $x$.
 Consequently, taking into account the fact that  $(\mathrm I -\beta  R_{\alpha ,0 }(z)|_{\mu \mu})^{-1}$ is bounded, the operator $G_{\alpha, \beta } (z)\,:\,L^2 _\mu  \to L^2_\mu$ can be estimated
    \begin{equation}\label{eq-asymG-gen}
\| G_{\alpha, \beta } (z)\|_{\mathrm{HS}} \leq C_1 \frac{ \e ^{-\sqrt{2}\mathcal T_z \rho } }{\rho}\,,
    \end{equation}
    where $C_1$ is a positive constant depending on $|\mu|$ and $z$. Formula~(\ref{eq-asymG-gen})
is a generalisation of (\ref{eq-asymG}).

Repeating  the reasoning from~Section~\ref{sec-analytic-extension}
we construct  the second sheet analytic continuation of $R_{\alpha, 0}(z)$ via the segment $(E_{\alpha; j}\,, E_{\alpha; j+1})$ in the sense of (\ref{eq-resolvent-alpha-beta-IIsheet}). Using this result and relying on the same analysis as in the  proof of~Theorem~\ref{th-resonances} we can prove the that $R_{\alpha, \beta \mu } (z)$ has the second sheet pole located at 
\begin{equation} \label{ieq-asympII}
   z_ n(\rho )= \mathcal E_{\beta , n}+\mathcal O \Big(\frac{\e^{-\sqrt{2 |\mathcal E_{\beta ;n}|} \rho}  }{\rho }\Big)\,.
\end{equation}
Moreover,  the dominated term of its imaginary part satisfies
 \begin{align} \label{eq-Gamman1-gen}
\begin{split} \Gamma _n (\rho )=  -\beta ^2 \sum_{k=1}^j\frac{\pi}{ 2 \sqrt{ \mathcal E_ {\beta;n }- E_{\alpha; k}  }
} & \Big [\Big| \int_{\Omega }I_\mu (\omega_ {\beta; n } (x)^{\ast } \varphi_{k} (x,p_1 ) )\mathrm d \mu ( x) \Big|^2 _{p_1=(\mathcal E_{\beta; n}- E_{\alpha; k } )^{1/2 }} \\ &+
\Big| \int_{\Omega }I_\mu (\omega_ {\beta; n } (x)^{\ast } \varphi_{k} (x,p_1 ) )\mathrm d \mu ( x) \Big|^2 _{p_1=-(\mathcal E_{\beta; n}- E_{\alpha; k } )^{1/2 }}\Big]\,.
\end{split}
  \end{align}
The above result states a generalization of  Fermi's golden role (\ref{eq-Gamman1}) for the Kato class measures. 

 \subsection{A special case: the trap defined by a delta interaction}

In this section we consider a special case of  the Kato class measure, defined by the Dirac delta potential supported on a curve.
Let $\Omega $ be a finite $C^1$ planar curve (open arc or  loop) with the length  $l$.
Assume that \(\Omega\) is parameterized by the arc length \(s\in I_l := [0, l]\) and is the graph of the function \( s \mapsto \varpi(s) \in \R^2\).
Our aim is to analyze the resonances problem for  the Hamiltonian which can be symbolically written
\begin{equation}\label{eq-delta-formal}
   -\Delta - \alpha \chi_\Sigma (x)  - \beta\delta (x - \Omega )\,,\quad \alpha <0 \,, \,\,\beta \in \mathbb R\,.
\end{equation}
\begin{remark} Delta interaction: a particular case of Kato class measure. \rm{
In the analyzed case, the measure \(\mu\) corresponding to the delta interaction will be denoted as $\delta $. It is defined as follows: for any Borel set \(\mathcal{B} \subset \mathbb{R}^2\), we determine
\begin{equation}\label{eq-lin}
\delta (\mathcal{B}) = \mathrm{lin}(\mathcal{B} \cap \Omega),
\end{equation}
where \(\mathrm{lin}(\cdot)\) denotes the linear measure. Consequently, \(\delta \) belongs to the Kato class of measures, cf.~\cite{BEKS}.
}
\end{remark}
To give the expression (\ref{eq-delta-formal}) a mathematical meaning, we introduce  the corresponding embedding
$I_\delta  \,:\, W^{1,2 } (\R^2 )\to L^2 (I_l)$ acting as $I_\delta  f = f(\varpi (\cdot ))$.
Then, for this specific case,
 the sesquilinear form (\ref{eq-formalphabeta}) can be expressed as
\begin{equation}\label{eq-formalphabeta-delta}
     \mathcalligra { E}_{\alpha, \beta \delta } [f]= \int_{\mathbb R^2 } |\nabla f (x) |^2 \,\mathrm d x- \alpha
     \int_{\Sigma }|f (x)|^2 \,\mathrm d x
-\beta \int_{I_l}|  f(\varpi (s)|^2 \,\mathrm d s\,, \quad f\in W^{1,2 } (\R^2 )\,.
\end{equation}
The special case of (\ref{eq-estImu2}) takes the form
\begin{equation}\label{eq-estImu3}
\int_{I_l} |f(\varpi (s))|^2 \mathrm d s\leq \varepsilon \Big[\int_{\mathbb R^2 } |\nabla f(x)|^2 \mathrm d x
-  \alpha \int_{\Sigma  } | f(x)|^2 \mathrm d x \Big] + C_1(\varepsilon ) \int_{\mathbb R^2 } | f(x)|^2 \mathrm d x\,.
\end{equation}
Consequently, the operator $ H_{\alpha, \beta \delta } $ associated to  $\mathcalligra { E}_{\alpha, \beta \delta }$ is self adjoint.
\begin{remark} \label{re-negative-beta} Negative eigenvalues of $H_{0, \beta \delta }$ for the delta interaction trap. \rm{
It is worth mentioning that for any negative coupling constant $\beta$, the operator $H_{0, \beta \delta }$ always admits at least one negative eigenvalue. However, it should be noted at the same time that in systems of dimensions three or higher, a sufficiently weak coupling constant supported on a finite curve does not generate a discrete spectrum,
cf.~\cite[Thm 4.2(iii)]{BEKS}.}
\end{remark}
The embeddings of  $R(z)$ to $ L^2(I_l)$ defined by the Dirac delta are given by
\begin{equation}\label{eq-embeddelta1}
 R(z)|_{0 \delta }f= \frac{1}{2\pi}\int_{\mathbb R^2 } K_0 (k_z|x- \varpi (s)|)f(s)\, \mathrm d s\,:\, L^2 (I_l) \to L^2 (\mathbb R^2)\,,
 \end{equation}
and 
\begin{equation}\label{eq-embeddelta2}R(z)|_{\delta \delta } :=
\frac{1}{2\pi}\int_{I_l} K_0 (k_z|\varpi (\cdot )- \varpi (s)|)f(s)\, \mathrm d s
\,:\, L^2 ( I_l)\to L^2 (I_l)\,,\end{equation}
moreover,
$ R(z)|_{ \delta  0 }= I_\delta  R(z) \,:\,L^2 (\mathbb R^2)\to L^2 (I_l)$.
This allows to construct the resolvent  $(H_{0, \beta \delta  }- z )^{-1}$ by
\begin{equation} \label{eq-betaresolvent1gen-delta}
   R_{0,\beta \delta  }(z)= R (z) +\beta  R(z)|_{0\delta } (\mathrm I -\beta  R(z)|_{\delta \delta })^{-1} R (z)|_{\delta  0} \,.
\end{equation}

Assume that  $\beta >0$. Relying on Remark~\ref{re-negative-beta} we conclude  that the set of negative eigenvalues of $H_{0,\beta \delta   }$ is not empty;  the same holds for $H_{\alpha ,0  }$. This means that the  assumption formulated for the general measure $\mu$ are automatically  satisfied. We keep the notation $\mathcal E_{\beta ; n}$, $n\in \mathcal N_\beta $ and  $\mathcal E_{\alpha ; k}$, $k\in \mathcal N_\alpha  $.
For the delta interaction the inequalities (\ref{eq-asymR0-gen})
take the forms
\begin{align}
\|  R (z)|_{\delta  0  } \chi_\Sigma \|^2_{\mathrm{HS}} &\leq
 \sup_{x\in \Omega }\frac{l}{2\pi }\int_{\Sigma }| K_0 (k_z |x- y|) |^2 \mathrm d y \leq C\frac{ l\ }{|z|^{1/2} \mathcal T_z}\, \frac{\e ^{-\sqrt{2}\mathcal T_z\rho } }{\rho}\,.
\end{align}
Furthermore, a realization of (\ref{eq-Gamman1-gen}) for the discussed special case is giveen by
 \begin{align} \label{eq-Gamman1-gen-delta}
\begin{split}\Gamma _n (\rho )=  -\beta ^2 \sum_{k=1}^j\frac{\pi}{ 2 \sqrt{ \mathcal E_ {\beta;n }- E_{\alpha; k}  }
} & \Big [\Big| \int_{I_l } (\omega_ {\beta; n } (\varpi (s))^{\ast } \varphi_{k} (\varpi (s)p_1 ) )\mathrm d s \Big|^2 _{p_1=(\mathcal E_{\beta; n}- E_{\alpha; k } )^{1/2 }} \\ &+
\Big| \int_{I_l }(\omega_ {\beta; n } (\varpi (s))^{\ast } \varphi_{k} (\varpi (s)),p_1 ) )\mathrm d s \Big|^2 _{p_1=-(\mathcal E_{\beta; n}- E_{\alpha; k } )^{1/2 }}\Big]\,.
 \end{split}   \end{align}
 The above formula provides Fermi's golden rule for the resonant pole that appears  in a neighborhood of \( \mathcal{E}_{\beta; n} \in (E_{\alpha; j}, E_{\alpha; j+1}) \), in systems governed by \( H_{\alpha, \beta \delta} \).
\begin{remark} Waveguide with a point interaction. \rm{
Notably, Fermi's golden rule does not hold in the form analogous to (\ref{eq-Gamman1-gen-delta}) for a trap defined by a more singular potential determined by a point interaction. This follows from the fact that a point potential in two dimensional system requires logarithmic regularization, cf.~\cite{AGHH}, and consequently, the eigenfunctions $\omega_{\beta;n}$ exhibit a corresponding singularity at this  point. While we can expect resonances with an exponentially decaying imaginary component, Fermi's golden rule requires modification. We defer the exploration of this topic, along with related concepts, to future research.}
\end{remark}

\section*{Appendix: The asymptotics of futher terms contributing  to $\Im z_n (\rho)$} \label{sec-appendix}

First, we discuss the asymptotics of $(w_{\beta ;n}\, , G_{\alpha, \beta }(z)  A_n (z)G_{\alpha, \beta } (z) w_{\beta ;n})_{L^2 (\Omega )}$, which is contributes to  (\ref{eq-expansion3}). Repeating the argument from the proof of Theorem~\ref{th-resonances} it suffices to consider this component for $z=\mathcal E_{\beta; n}$. Furthermore, since we  are interested Fermi's golden rule
our aim is to analyze
\begin{eqnarray} \label{eq-asympIII}
\Im (w_{\beta ;n}\, , G_{\alpha, \beta }(\mathcal E_{\beta; n})  A_n (\mathcal E_{\beta; n})G_{\alpha, \beta } (\mathcal E_{\beta; n}) w_{\beta ;n})_{L^2 (\Omega )}\,.
\end{eqnarray}
We use the decomposition $
G_{\alpha, \beta }(\mathcal E_{\beta; n})  = G_{\alpha, \beta }^R (\mathcal E_{\beta; n}) +iG_{\alpha, \beta }^I(\mathcal E_{\beta; n}) $, where $G_{\alpha, \beta }^R (\mathcal E_{\beta; n})$ and $G_{\alpha, \beta }^I(\mathcal E_{\beta; n})$ are self adjoint. In fact $G^I_{\alpha, \beta }$ is determined by the operator employed in (\ref{eq-imresolvent}). Then (\ref{eq-asympIII}) can be written
\begin{eqnarray} \begin{split} & \Im (w_{\beta ;n} ,G_{\alpha, \beta }(\mathcal E_{\beta; n})  A_n (\mathcal E_{\beta; n})G_{\alpha, \beta } (\mathcal E_{\beta; n}) w_{\beta ;n})_{L^2 (\Omega )} \\  & =
(w_{\beta ;n}\, , G^I_{\alpha, \beta }(\mathcal E_{\beta; n})  A_n (\mathcal E_{\beta; n})G_{\alpha, \beta } ^R(\mathcal E_{\beta; n}) w_{\beta ;n})_{L^2 (\Omega )}+(w_{\beta ;n}\, , G^R_{\alpha, \beta }(\mathcal E_{\beta; n})  A_n (\mathcal E_{\beta; n})G_{\alpha, \beta } ^I(\mathcal E_{\beta; n}) w_{\beta ;n})_{L^2 (\Omega )}\,.\end{split}
\end{eqnarray}
It suffices to estimate one of the above components. Using the expression from (\ref{eq-imresolvent}) we get
\begin{align}
\begin{split}
\big|\Im  (\chi_\Omega \omega_{\beta ; n }& , G^I_{\alpha, \beta }(\mathcal E_{\beta; n})  A_n (\mathcal E_{\beta; n})G_{\alpha, \beta } ^R(\mathcal E_{\beta; n})\chi_\Omega \omega_{\beta ; n })_{L^2 (\Omega )}\big| \\  &\leq
\sum_{k=1} ^j
\frac{\beta }{  2\sqrt{ \mathcal E_ {\beta;n }- E_{\alpha; k}  } }
\int_{\Omega \times \Omega  } \big|\cos ( \sqrt{ \mathcal E_ {\beta;n }- E_{\alpha; k}  }|x_1 - y_1|)
 \phi_{\alpha;k }(x_2) \phi_{\alpha;k}(y_2)^\ast \omega_{\beta ; n } (x)^\ast \\   &\quad \qquad  \times ( A_n (\mathcal E_{\beta; n})G_{\alpha, \beta } ^R(\mathcal E_{\beta; n})\omega _{\beta ; n } )(y) \big|\,\mathrm d x \,\mathrm d y
 \\   & \leq  \mathrm{const}
\, \mathrm{e}^{-2\sqrt{ |E_{\alpha , j }|}\rho } \,|\Omega |^2 \, \|\omega _{\beta ; n } \|_{L^1 (\Omega )}\|A_n (\mathcal E_{\beta; n})G_{\alpha, \beta } ^I(\mathcal E_{\beta; n})\omega _{\beta ; n } \|_{L^1 (\Omega )}
 \\  & \leq  \mathrm{const}
\, \mathrm{e}^{-2\sqrt{| E_{\alpha , j }|}\rho }  \,
|\Omega |^2 \, \|\omega _{\beta ; n } \|_{L^2 (\Omega )} \|A_n (\mathcal E_{\beta; n})G_{\alpha, \beta } ^I(\mathcal E_{\beta; n})\omega _{\beta ; n } \|_{L^2(\Omega )}\,,
 \\   & \leq  \mathrm{const}
  \,
\frac{\mathrm{e}^{-2\sqrt{| E_{\alpha , j }|}\rho }\,\e^{-\sqrt{2 |\mathcal E_{\beta ;n}|} \rho}  }{\rho }
|\Omega |^2 \, \|\omega _{\beta ; n } \|_{L^2 (\Omega )}^2\,,
\end{split}
 \end{align}
where we used   $\|f \|_{L^1(\Omega )}\leq |\Omega |\, \|f \|_{L^2(\Omega )}$ for  $f\in L^2 (\Omega )$, in the last inequality we applied the norm bound for $G_{\alpha , \beta } (\cdot )$, (see~(\ref{eq-asymG})) and the boundness of $A_n (\mathcal E_{\beta ;n})$; the constants contained in the symbol '\rm{const}' depend on  $j$ and coupling constants $\alpha $ and $\beta$. The resulting estimate shows that  the asymptotics of (\ref{eq-asympIII}) is stronger then (\ref{eq-asympII}) and, consequently, is does not contribute to the dominated term. The further terms contributing to (\ref{eq-expansion3})  can be estimated analogously.

\subsection*{Acknowledgement}
The authors would like to thank the anonymous referees for the careful reading of the paper and for their valuable remarks and comments, which have enhanced its quality.
This work was partially supported by a program of the Polish Ministry of Science under the title ‘Regional Excellence Initiative’, project no. RID/SP/0050/2024/1.



\section*{Data Availability Statement}
Data sharing is not applicable to this article as no new data were created or analyzed in this study.

\bigskip
\section*{Declaration of conflicting interests.} Authors declare that they have no conflicts of interest.



\begin{thebibliography}{99}

\bibitem{AS} M.~Abramowitz and I.~Stegun, Handbook of Mathematical Functions with Formulas, Graphs, and Mathematical Tables.


\bibitem{AF2003} R.~A.~Adams, J.~J.~F.~Fournier, Sobolev Spaces,
2nd Edition, Volume 140 - June 26, 2003.

\bibitem{AGHH} S.~Albeverio, F.~Gesztesy, R.~Hoegh-Krohn and  H.~Holden, with an appendix by P.~Exner,  {\it Solvable Models in Quantum Mechanics} Second Edition, AMS Chelsea Publishing, Providence, R.I., 2005.

\bibitem{BFRL}
V.~Barrera-Figueroa, V.~Rabinovich and S.~Loredo-Ramírez,
Asymptotic and numerical analysis of slowly varying two-dimensional quantum waveguides,  {\it J. Phys. A: Math. Theor.}  {\bf 55} (2022)
095202.

\bibitem{BG}	
P.~Briet and M.~Gharsalli,
Stark resonances in a quantum waveguide with analytic curvature,  {\it J. Phys. A: Math. Theor.}   {\bf 49} (2016) 495202.


\bibitem{BEG} J.~Behrndt, A.~F.~M. ter Elst and F.~Gesztesy, The generalized Birman–-Schwinger principle {\it
Trans. Amer. Math. Soc.} {\bf 375} (2022), 799-845.

\bibitem{BEKS} J. F. Brasche, P. Exner, Yu. A. Kuperin and P. \v{S}eba, Schr\"odinger operators
with singular interactions, {\it J. Math. Anal. Appl.} {\bf 184} (1994), 112–139. 

\bibitem{CAZEG} P.~L.~Christiansen, H.~C.~Arnbak, A.~V.~Zolotaryuk, V.~N.~Ermakov and Y.~B.~Gaididei,
On the existence of resonances in the transmission probability for interactions arising from derivatives of Dirac's delta function,
{\it J. Phys. A: Math. Gen. }
{\bf 36}   (2003), 7589.

\bibitem{Dirac}  P.~A.~M.~Dirac, The Quantum Theory of Emission and Absorption of Radiation, {\it Proc. Roy. Soc.} {\bf A}114 243 (1927), 243–265.


\bibitem{DeGr} A.~Delitsyn and  D.~S.~Grebenkov, Resonance scattering in a waveguide with identical thick
perforated barriers, {\it Applied Mathematics and Computation}, {\bf 412} (2022) pp.126592.



\bibitem{DEM}
P. Duclos, P. Exner and B. Meller,   Exponential bounds on curvature induced resonances in
a two-dimensional Dirichlet tube, {\it Helv. Phys. Acta}  {\bf 71} (1998), 477-492.



\bibitem{EK2015} P. Exner and H. Kova\v{r}\'{i}k, Quantum Waveguides, 2015.


\bibitem{Exner21} P.~Exner, Spectral properties of soft quantum waveguides
 {\it J. Phys. A: Math. Theor.} {\bf 53} (2020),  No.~355302.

\bibitem{Exner22} P. Exner, Soft quantum waveguides in three dimensions, {\it J.Math.Phys.} {\bf 63}(4) (2022), 042103.

\bibitem{EKL23} P. Exner, S. Kondej and  V. Lotoreichik, Bound states of weakly deformed
soft waveguides. Accepted for {\it Asymp. Anal.} Preprint on arXiv:2211.01989 [math-ph].

\bibitem{ExnerKondej04} P. Exner and  S. Kondej, Schr\"odinger operators with singular interactions: a model of tunneling resonances, {\it J.~Phys.~A: Math.~and~Gen.} {\bf 37}(34)  (2004), 8255.

\bibitem{ExnerLotoreichik} P. Exner and V. Lotoreichik, Optimization of the lowest eigenvalue of a soft quantum ring,
{\it Lett. Math. Phys.} {\bf 111}(28) (2021).

\bibitem{ExnerS2023} P.~Exner and D.~Spitzkopf, Tunneling in soft waveguides:closing a book, {\it J. Phys. A: Math. Theor.} {\bf 57}
(2024) 125301.

\bibitem{ExnerV} P. Exner and  S. Vugalter,  Bound states in bent soft waveguides, accepted for {\it J.~Spect.~Th.},
Preprint on arXiv:2304.14776 [math-ph].

\bibitem{Fermi} E.~Fermi, Nuclear Physics University of Chicago Press 1950.

\bibitem{GHN2015}F.~Gesztesy, H.~Holden and  R.~Nichols 
On Factorizations of Analytic Operator-Valued Functions and Eigenvalue Multiplicity Questions {\it Integr. Equ. Oper. Theory } {\bf 8} (2015) 61–94.


\bibitem{Teufel2} S. Haag, J. Lampart and  S. Teufel,  Generalised Quantum Waveguides
{\it Annales Henri Poincaré} {\bf 16}, (2015) 2535–2568.

\bibitem{HamsmanKrejcirik} M. Hansmann and D. Krej\v{c}i\v{r}\'{i}k, The abstract Birman--Schwinger principle and spectral stability,
{\it J. Anal. Math.} {\bf 148} (2022) 361-398.







\bibitem{Kato} T.~Kato, {\it Perturbation Theory for Linear Operators}, Springer, 1995.


\bibitem{KSJC} C.~S.~Kim, A.~M.~Satanin
Y.~S.~Joe and R.~M.~Cosby, Resonant tunneling in a quantum waveguide: Effect of a finite-size attractive impurity,
Phys.~Rev.~B {\bf 60}(15) (1999).

\bibitem{Kondej2024} S. Kondej, Quantum soft waveguides with resonances induced by broken symmetry.  {\it J.  Phys. A: Math. Theor.} {\bf  57} (2024) 195201.


\bibitem{KondejKK21} S. Kondej, D. Krej\v{c}i\v{r}\'{i}k and  J. K\v{r}\'{i}\v{z}, Soft quantum waveguides with an explicit cut-locus,
{\it J. Phys. A: Math. Theor.} {\bf 54} (2021).


\bibitem{KL} S.~Kondej, V.~Lotoreichik, 2D Schrodinger operators perturbed by finite measures,
 {\it J. Math.~Anal.~Appl.} {\bf 420}(2), (2014).

\bibitem{KovarikS} H.~Kovařík and A.~Sacchetti, Resonances in twisted quantum waveguides, {\it J. Phys. A: Math. Theor.}
 {\bf 40}, 29  (2007) 8371.


\bibitem{KK22} D. Krej\v{c}i\v{r}\'{i}k, J. K\v{r}\'{i}\v{z}, Bound states in soft quantum layers
{\it Publ. RIMS, Kyoto University}, {\bf 60} (2024) 4.

\bibitem{LipLot} J.~Lipovský and V.~Lotoreichik, Asymptotics of Resonances Induced by Point Interactions
{\it  Acta.~Phys.~Pol.~A} {\bf 132} (6) (2017).






\bibitem{Po1} A.~Posilicano,
A Krein-like Formula for Singular Perturbations of Self-Adjoint Operators and Applications,
{\it J.~Func.~Anal.}
{\bf 183} (2001) 1, 109-147.

\bibitem{RS4} M. Reed, B. Simon. \textit{ Methods of Modern Mathematical
Physics. IV. Analysis of Operators}, Academic Press, New York,
1978.

\bibitem{RS1} M.~Reed and B. Simon.  {\it Methods of Modern Mathematical
Physics, vol. I. Functional analysis}, Academic Press, New
York, 1980.


\bibitem{Simon-weakly} B.~Simon, The bound state of weakly coupled Schr\"odinger operators in one and two dimensions, {\it Ann. of Phys.}
{\bf 97}, Issue 2, (1976),  279-288.

\bibitem{Simon-comprehensive} B.~Simon, \textit{ Operator Theory. A Comprehensive Course in Analysis, Part 4}, American Mathematical Society, Providence,
2015.
\end{thebibliography}
\end{document}